\documentclass[12pt]{article}
\usepackage{amssymb}
\usepackage{amsmath}
\usepackage{color}
\usepackage[unicode,bookmarks,bookmarksopen,bookmarksopenlevel=2,colorlinks,linkcolor=blue,citecolor=green]{hyperref}

\newtheorem{theorem}{Theorem}

\newtheorem{corollary}[theorem]{Corollary}

\newtheorem{proposition}[theorem]{Proposition}
\newtheorem{remark}[theorem]{Remark}

\newenvironment{proof}[1][Proof]{\textbf{#1.} }{\ \rule{0.5em}{0.5em}}

\newcommand*{\fd}
[2]{\mathchoice{\frac{\delta#1}{\delta#2}}
  {\delta #1/\delta#2}{\delta#1/\delta#2}{\delta#1/\delta#2}}
\newcommand{\ddx}[1]{\partial_x^{#1}}
\input{tcilatex}
\begin{document}

\title{\textbf{On the bi-Hamiltonian Geometry\\
of WDVV Equations}}
\author{Maxim V. Pavlov$^{1,2}$, Raffaele F. Vitolo$^{3}$ \\
%EndAName
$^{1}$Department of Mathematical Physics,\\
Lebedev Physical Institute of Russian Academy of Sciences,\\
Leninskij Prospekt, 53, 119991 Moscow, Russia\\
\texttt{m.v.pavlov@lboro.ac.uk}\\
[3mm] $^{2}$Department of Applied Mathematics,\\
National Research Nuclear University MEPHI,\\
Kashirskoe Shosse 31, 115409 Moscow, Russia\\
[3mm] $^{3}$Department of Mathematics and Physics \textquotedblleft E. De
Giorgi\textquotedblright ,\\
University of Salento, Lecce, Italy\\
\texttt{raffaele.vitolo@unisalento.it} }
\date{}
\maketitle

\begin{abstract}
We consider the WDVV associativity equations in the four dimensional case.
These nonlinear equations of third order can be written as a pair of six
component commuting two-dimensional non-diagonalizable hydrodynamic type
systems. We prove that these systems possess a compatible pair of local
homogeneous Hamiltonian structures of Dubrovin--Novikov type (of first and
third order, respectively).

\bigskip


\bigskip

\end{abstract}

% \tableofcontents

\section{Introduction}

The Witten--Dijkgraaf--Verlinde--Verlinde (WDVV) associativity equations arise
as the conditions of associativity of an algebra in an $N$ dimensional
space. The mathematical theory of the equations can be found in \cite{Dub2b},
as well as references to papers with the physical motivations. The equations
are a system of third-order PDEs in one unknown function $%
F=F(t^{1},\dots ,t^{N})$. Namely, it is assumed that
\begin{equation*}
\eta _{\alpha \beta }=\frac{\partial ^{3}F}{\partial t^{1}\partial t^{\alpha
}\partial t^{\beta }}
\end{equation*}%
is a constant nondegenerate symmetric matrix ($\eta ^{\alpha \beta }$ will
denote its inverse matrix); the WDVV associativity equations are equivalent
to the requirement that the functions%
\begin{equation*}
c_{\beta \gamma }^{\alpha }=\eta ^{\alpha \mu }\frac{\partial ^{3}F}{%
\partial t^{\mu }\partial t^{\beta }\partial t^{\gamma }}
\end{equation*}%
are the structure constants of an associative algebra. Then the
associativity condition reads as%
\begin{equation*}
\eta ^{\mu \lambda }\frac{\partial ^{3}F}{\partial t^{\lambda }\partial
t^{\alpha }\partial t^{\beta }}\frac{\partial ^{3}F}{\partial t^{\nu
}\partial t^{\mu }\partial t^{\gamma }}=\eta ^{\mu \lambda }\frac{\partial
^{3}F}{\partial t^{\nu }\partial t^{\alpha }\partial t^{\mu }}\frac{\partial
^{3}F}{\partial t^{\lambda }\partial t^{\beta }\partial t^{\gamma }}
\end{equation*}

The integrability of the above equations was proved in \cite{Dub2b} by giving a
Lax pair for all values of $N$ and $\eta _{\alpha \beta }$. The Hamiltonian
geometry of WDVV associativity equations also attracted interest of a number of
researchers. In particular, a fundamental contribution was given in papers
\cite{FM,FGMN}, where the case $N=3$ was considered. When $N=3$ we have the
first nontrivial case with just one WDVV associativity equation. If the matrix
$\mathbf{\eta }$ is antidiagonal, \emph{i.e.} $\eta _{\alpha \beta }=\delta
_{\alpha +\beta ,4}$ and $F=\frac{1}{2}%
(t^{1})^{2}t^{3}+\frac{1}{2}t^{1}(t^{2})^{2}+f(t^{2},t^{3})$, the WDVV
associativity equation is (after setting $x=t^{2}$, $t=t^{3}$)
\begin{equation}
f_{ttt}=f_{xxt}^{2}-f_{xxx}f_{xtt}.  \label{eq:4}
\end{equation}%
Introducing the new variables $a^{1}=a=f_{xxx}$, $%
a^{2}=b=f_{xxt} $, $a^{3}=c=f_{xtt}$, the WDVV associativity equation can be
written as an hydrodynamic type system of PDEs
\begin{equation}
a_{t}^{i}=v_{j}^{i}(\mathbf{a})a_{x}^{j};  \label{eq:5}
\end{equation}%
more precisely, $a_{t}=b_{x}$, $b_{t}=c_{x}$, $c_{t}=(b^{2}-ac)_{x}$. It was
proved in \cite{FGMN} that the above system can be rewritten as a
bi-Hamiltonian system 
\begin{equation}
a_{t}^{i}=A_{1}^{ij}\fd{H_2}{a^j}=A_{2}^{ij}\fd{H_{1}}{a^j}.  \label{eq:6}
\end{equation}%
with respect to two compatible local Hamiltonian operators $\hat{A}_{1}$ and 
$\hat{A}_{2}$, with expressions 
\begin{gather*}
\hat{A}_{1}=%
\begin{pmatrix}
-\frac{3}{2}\partial _{x}^{{}} & \frac{1}{2}\partial _{x}^{{}}a & \partial
_{x}^{{}}b \\ 
\frac{1}{2}a\partial _{x}^{{}} & \frac{1}{2}(\partial _{x}^{{}}b+b\partial
_{x}^{{}}) & \frac{3}{2}c\partial _{x}^{{}}+c_{x} \\ 
b\partial _{x}^{{}} & \frac{3}{2}\partial _{x}^{{}}c-c_{x} & 
(b^{2}-ac)\partial _{x}^{{}}+\partial _{x}^{{}}(b^{2}-ac)%
\end{pmatrix}
\\
\hat{A}_{2}=%
\begin{pmatrix}
0 & 0 & \partial _{x}^{3} \\ 
0 & \partial _{x}^{3} & -\partial _{x}^{2}a\partial _{x} \\ 
\partial _{x}^{3} & -\partial _{x}a\partial _{x}^{2} & \partial
_{x}^{2}b\partial _{x}+\partial _{x}b\partial _{x}^{2}+\partial
_{x}a\partial _{x}a\partial _{x}%
\end{pmatrix}%
\end{gather*}%
and Hamiltonian densities, respectively, $h_{2}=c$, $h_{1}=-\frac{1}{2}%
a(\partial _{x}^{-1}b)^{2}-(\partial _{x}^{-1}b)(\partial _{x}^{-1}c)$,
where $H_{i}=\int h_{i}dx$. Here by a `Hamiltonian operator' we mean an
operator $\hat{A}$ such that its Schouten bracket $[\hat{A},\hat{A}]$
vanishes, and by `compatible' (or commuting) operators we mean that the
Schouten bracket $[\hat{A}_{1},\hat{A}_{2}]$ vanishes (see \cite{Dorf} for
the definition).

The two Hamiltonian operators $\hat{A}_{1}$ and $\hat{A}_{2}$ are members of
a class of operators which has been introduced by Dubrovin and Novikov \cite%
{DN,DN2}. They are homogeneous with respect to the grading $\deg a^{i}=0$, $%
\deg \partial _{x}^{{}}=1$. They have interesting geometric properties which
were completely described in \cite{DN} for first order Hamiltonian
operators. Third order Hamiltonian operators have a more complicated
structure (see \cite{GP87,GP97,GP91,BP,Doyle}). A complete
classification of these Hamiltonian operators was found just in two and
three component cases (see detail in \cite{fpv}).

In particular, $\hat{A}_{1}$ was found in \cite{FM}, and it is
completely specified by a contravariant flat pseudo-Riemannian metric $%
g^{ik}$. The observation that led to finding $\hat{A}_{1}$ was that the
eigenvalues $u^{k}(\mathbf{a})$ of one of the matrices of the Lax pair of
the system \eqref{eq:5} are conservation law densities. If the system is
rewritten using the above eigenvalues as new dependent variables $u^{k}$,
the Hamiltonian operator $\hat{A}_{1}$ becomes evident and is of the type $%
A_{1}^{ij}=K^{ij}\partial _{x}^{{}}$, where $K^{ij}$ is a constant symmetric
non-degenerate matrix. Hamiltonian operators of this type are said to be 
\emph{hydrodynamic type Hamiltonian operators}.

The Hamiltonian operator $\hat{A}_{2}$ was found in a completely different
way. More precisely, a Lagrangian for the $x$-derivative of the WDVV
associativity equation \eqref{eq:4} was found, and a symplectic
representation of this equation was achieved in \cite{FGMN}. Then $\hat{A}%
_{2}$ was found by inverting the corresponding symplectic form and
multiplying it by $\hat{A}_{1}$. It is necessary to emphasize that the
coordinates $a^{k}$ (see \eqref{eq:5}) are Casimirs densities for $\hat{A}%
_{2}$, \emph{i.e.} $\hat{A}_{2}$ has vanishing free term in these
coordinates. It is known \cite{GP87,GP97,GP91} that $\hat{A}_{2}$ has a
particularly simple form \eqref{casimir} when written with respect to its
Casimirs densities; in these coordinates $a^{k}$ the inverse matrix of the
leading coefficient $g^{ik}$ is a Monge metric \cite{fpv}, which in this
example reads as 
\begin{equation*}
g_{ij}=%
\begin{pmatrix}
-2b & a & 1 \\ 
a & 1 & 0 \\ 
1 & 0 & 0%
\end{pmatrix}%
\end{equation*}

Other approaches to the Hamiltonian geometry of WDVV equations appeared in the
literature in the case $N=3$. For instance, in \cite{KN1} a different
identification of the variables $t^{2}$ and $t^{3}$ as $t$ and $x$ leads to
different WDVV associativity equations and different bi-Hamiltonian
formulations through local Dubrovin--Novikov operators of the first and third
order. Moreover, in \cite{KN2} another choice of constants in $\eta _{\alpha
  \beta }$ was investigated. Its third order Hamiltonian operator lies in a
different class (cf. \cite{FGMN}) with respect to the classification in
\cite{fpv}. Finally, in \cite{KeKrVeVi-TMF-2010} Hamiltonian operators for the
WDVV equation in the non-evolutionary form~\eqref{eq:4} have been found
(although with an explicit dependency on the independent variables).

In the case $N=4$ the bi-Hamiltonian nature of WDVV associativity equations was
an open problem. If the matrix $\mathbf{\eta }$ is antidiagonal\footnote{%
  The constant symmetric matrix $\mathbf{\eta}$ can always be reduced by a
  linear change of the coordinates $(t^i)$ to either the antidiagonal form if
  $\eta_{11}=0$, or to another form if $\eta_{11}\neq 0$. Only the first case
  admits physically relevant examples.  See \cite{Dub2b} for details.},
\emph{i.e.} the choice $\eta _{\alpha \beta }=\delta _{\alpha +\beta ,5}$, WDVV
associativity equations were considered in \cite{Dub2b}. This implies $F=
\frac{1}{2}(t^{1})^{2}t^{4}+t^{1}t^{2}t^{3}+f(t^{2},t^{3},t^{4})$. Then we have
the following WDVV associativity equations (after the identifications $%
x=t^{2}$, $y=t^{3}$, $z=t^{4}$)
\begin{equation}
\begin{split}
& -2f_{xyz}-f_{xyy}f_{xxy}+f_{yyy}f_{xxx}=0, \\
& -f_{xzz}-f_{xyy}f_{xxz}+f_{yyz}f_{xxx}=0, \\
& -2f_{xyz}f_{xxz}+f_{xzz}f_{xxy}+f_{yzz}f_{xxx}=0, \\
& -f_{yyy}f_{xxz}+f_{yzz}+f_{yyz}f_{xxy}=0, \\
& f_{zzz}-(f_{xyz})^{2}+f_{xzz}f_{xyy}-f_{yyz}f_{xxz}+f_{yzz}f_{xxy}=0, \\
& f_{yyy}f_{xzz}-2f_{yyz}f_{xyz}+f_{yzz}f_{xyy}=0.
\end{split}
\label{eq:9}
\end{equation}

In \cite{FM6} (see also \cite{OM98}) it was found that the above system can
be rewritten as a compatible pair of two hydrodynamic type systems (see
Section \ref{sec:wdvv-as-commuting}). Such systems admit a Dubrovin--Novikov
type first-order Hamiltonian operator \cite{FM6}, but a second
Dubrovin--Novikov type third order Hamiltonian operator for these systems
was never found until now (see basic facts in \ref{sec:struct-second-hamilt}%
).

The main task of this paper is to completely uncover the bi-Hamiltonian
geometry of the WDVV associativity equations for $N=4$ and the antidiagonal
metric $\eta _{\alpha \beta }=\delta _{\alpha +\beta ,5}$ by finding a
second Dubrovin--Novikov type third order Hamiltonian operator for the
corresponding hydrodynamic-type systems. In order to achieve this goal, we
develop a procedure that can be fruitfully used to find such Hamiltonian
operators for other WDVV associativity systems, for example for different
matrices $\mathbf{\eta }$ or different values of $N$.

Below, we outline our procedure for convenience of the reader.
\begin{enumerate}
\item We start from the following data: the Lax pair and the first-order
  Hamiltonian structure $\hat{A}_1$.
\item In flat coordinates $(\mathbf{u})$ of $\hat{A}_1$, using the Lax pair, it
  is possible to find a sequence of homogeneous conservation law densities in
  the most compact form.
\item If we suppose that the above densities are related by a bi-Hamiltonian
  recursion with an unknown third-order operator $\hat{A}_2$ of
  Dubrovin--Novikov type, we can find a candidate to be the metric
  $g_{ij}(\mathbf{u})$ that is the coefficient of $\ddx{3}$ (this is the
  content of Theorem~\ref{th:reconstruction}).
\item We make the hypothesis that the Casimirs of $\hat{A}_2$ are the original
  coordinates $(\mathbf{a})$ in which the hydrodynamic-type WDVV systems are
  written. Indeed, this is the case: if we change coordinates from
  $(\mathbf{u})$ to $(\mathbf{a})$ we can prove that $g_{ij}(\mathbf{a})$ is a
  Monge metric and completely determines a Hamiltonian operator $\hat{A}_2$
  which is compatible with $\hat{A}_1$, according with the general theory in
  \cite{fpv} (Theorem~\ref{th:second-hamilt-struct}).
\item  In the same Theorem we also exhibit a factorization of the
Hamiltonian operator whose theoretical existence was ensured in \cite{Doyle}.
\item A deeper analysis of the properties of the factorization of the
  third-order Hamiltonian operator of Dubrovin--Novikov type performed in
  Section~ \ref{sec:potemin} applied to the WDVV hydrodynamic type systems
  allows us to formulate general existence theorems for nonlocal Casimirs of
  the Dubrovin--Novikov type third-order Hamiltonian operator, Hamiltonians for
  the WDVV hydrodynamic type systems and momentum, thus completing a
  bi-Hamiltonian picture of the WDVV hydrodynamic type systems in six
  components.
\end{enumerate}

The above procedure can be fruitfully used not only for WDVV
associativity equations but for any system of PDEs in $1+1$ dimensions.

We stress that the reconstruction of the Dubrovin--Novikov type third-order
Hamiltonian operator for the six component systems involves several steps
where coordinate expressions of the objects involved can only be handled by
computer algebra systems, and too big to be written down in this paper,
However, the final result, thanks to the factorized form that we achieve in
Theorem~\ref{th:second-hamilt-struct}, is so compact that Hamiltonian
properties of the WDVV hydrodynamic type systems can be checked by pen and
paper.

In particular, almost all symbolic computations were performed by CDE
\cite{cde}, a REDUCE package for integrability of PDEs. Only the
linearization of the WDVV hydrodynamic type systems and
their formal adjoint have been done with Jets \cite{Jets}, a Maple package
for the geometry of PDEs, since this feature will only appear
in the forthcoming version of CDE. The computation of homogeneous conservation
law densities required about 6 hours and 18GB of RAM on the server
\texttt{sophus} of the Dipartimento di Matematica e Fisica \textquotedblleft
E. De Giorgi\textquotedblright\ of the Universit\`{a} del Salento.

We stress that despite the fact that we heavily relied on computer algebra
calculations, our main results are checkable by pen and paper, in particular
the Hamiltonian property (with respect to the newly found third order operator
$\hat{A}_2$) of the $6$-component WDVV hydrodynamic type systems; see
Remark~\ref{sec:second-hamilt-struct-1}.

At the end of the paper a conclusive section contains several interesting
remarks about the perspectives of mathematical research on the bi-Hamiltonian
geometry of WDVV associativity equations. Here we would like to stress that at
this point it is natural to conjecture that every WDVV system admits a
third-order local Hamiltonian operator of Dubrovin--Novikov type, with a
distinguished subset of WDVV systems admitting a bi-Hamiltonian
formulation. Note that a Dubrovin--Novikov type third order Hamiltonian
operator is completely specified by an object from projective geometry, a
quadratic line complex \cite{Dolgachev,fpv}, which could be of interest in the
rich geometric framework that surrounds WDVV associativity equations.

\section{The WDVV Associativity\ Equations in $4$ Dimensions}

\label{sec:wdvv-as-commuting}

In this section we will give a brief summary of what is known from the
previous investigations in \cite{FM6} (see also \cite{OM98}).

For general $N$ the WDVV associativity equations can be presented as $N-2$
commuting two-dimensional non-diagonalizable hydrodynamic type systems with $%
n=N(N-1)/2$ components. The procedure for obtaining such systems is shown
here in details in the case $N=4$ \eqref{eq:9}. We call them WDVV
hydrodynamic type systems here and below.

We introduce new field variables $a^{k}$ in correspondence with every
derivative $f_{t^{i}t^{j}t^{k}}$ which contains at least one instance of $%
x=t^{2}$, i.e. $%
a^{1}=f_{xxx},a^{2}=f_{xxy},a^{3}=f_{xxz},a^{4}=f_{xyy},a^{5}=f_{xyz},a^{6}=f_{xzz} 
$. Then, the compatibility conditions for the WDVV associativity equations %
\eqref{eq:9} can be written as a pair of hydrodynamic type systems (cf. %
\eqref{eq:5})%
\begin{equation*}
a_{y}^{i}=v_{j}^{i}(\mathbf{a})a_{x}^{j},\text{ \ }a_{z}^{i}=w_{j}^{i}(%
\mathbf{a})a_{x}^{j};
\end{equation*}%
more precisely 
\begin{equation}
a_{y}^{i}=(v^{i}(\mathbf{a}))_{x},\text{ \ }a_{z}^{i}=(w^{i}(\mathbf{a}%
))_{x},  \label{eq:10}
\end{equation}%
where%
\begin{equation*}
v^{1}=a^{2},\text{ \ }w^{1}=a^{3},\text{ \ }v^{2}=a^{4},\text{ \ }%
v^{3}=w^{2}=a^{5},\text{ \ }w^{3}=a^{6},\text{ \ }v^{4}=f_{yyy}=\frac{%
2a^{5}+a^{2}a^{4}}{a^{1}},\text{ \ }v^{5}=w^{4}=f_{yyz}=\frac{%
a^{3}a^{4}+a^{6}}{a^{1}},
\end{equation*}%
\begin{equation*}
v^{6}=w^{5}=f_{yzz}=\frac{2a^{3}a^{5}-a^{2}a^{6}}{a^{1}},\text{ \ \ }%
w^{6}=f_{zzz}=(a^{5})^{2}-a^{4}a^{6}+\frac{%
(a^{3})^{2}a^{4}+a^{3}a^{6}-2a^{2}a^{3}a^{5}+(a^{2})^{2}a^{6}}{a^{1}}.
\end{equation*}

The commuting hydrodynamic type systems \eqref{eq:10} can be expressed as
the compatibility condition of two Lax pairs with a common part \cite{FM6}.
Such a common part takes the form%
\begin{equation}
\begin{pmatrix}
\psi \\ 
\psi _{1} \\ 
\psi _{2} \\ 
\psi _{3}%
\end{pmatrix}%
_{x}=\lambda \mathbf{A}%
\begin{pmatrix}
\psi \\ 
\psi _{1} \\ 
\psi _{2} \\ 
\psi _{3}%
\end{pmatrix}%
=\lambda 
\begin{pmatrix}
0 & 1 & 0 & 0 \\ 
a^{3} & a^{2} & a^{1} & 0 \\ 
a^{5} & a^{4} & a^{2} & 1 \\ 
a^{6} & a^{5} & a^{3} & 0%
\end{pmatrix}%
\begin{pmatrix}
\psi \\ 
\psi _{1} \\ 
\psi _{2} \\ 
\psi _{3}%
\end{pmatrix}%
.  \label{eq:11}
\end{equation}%
The characteristic equation of $\mathbf{A}$ is 
\begin{equation}
\det (\mathbf{A}-\rho \mathbf{I})=\rho ^{4}-2a^{2}\rho
^{3}+[(a^{2})^{2}-2a^{3}-a^{1}a^{4}]\rho ^{2}+2(a^{2}a^{3}-a^{1}a^{5})\rho
+(a^{3})^{2}-a^{1}a^{6}=0.  \label{eq:12}
\end{equation}%
It can be proved that its roots $u^{1}$, $u^{2}$, $u^{3}$, $u^{4}$ are
conservation law densities of both systems \eqref{eq:10}. The relation
between these roots and the coefficients of the characteristic equation is
given by the Vi\`{e}te formulae

\begin{equation}
\begin{split}
a^{2}& =\frac{1}{2}(u^{1}+u^{2}+u^{3}+u^{4}), \\
a^{3}& =\frac{1}{4}[(u^{1})^{2}+(u^{2})^{2}+(u^{3})^{2}+(u^{4})^{2}]-\frac{1%
}{8}(u^{1}+u^{2}+u^{3}+u^{4})^{2}-\frac{1}{2}a^{1}a^{4}, \\
a^{5}& =\frac{1}{2a^{1}}%
(2a^{2}a^{3}+u^{1}u^{2}u^{3}+u^{1}u^{2}u^{4}+u^{1}u^{3}u^{4}+u^{2}u^{3}u^{4}),
\\
a^{6}& =\frac{1}{a^{1}}[(a^{3})^{2}-u^{1}u^{2}u^{3}u^{4}].
\end{split}
\label{eq:13}
\end{equation}%
We can change the variables $a^{k}$ of the systems \eqref{eq:10} to the new
variables $u(\mathbf{a})$, i.e. $u^{0}=a^{1}$, $u^{1}$, $u^{2}$, $u^{3}$, $%
u^{4}$, $u^{5}=a^{4}$. These conservation law densities are flat coordinates
(see details in \cite{FM6}), i.e. hydrodynamic type systems \eqref{eq:10}
can be equipped by the Hamiltonian operator $\hat{A}_{1}=\mathbf{K}\partial
_{x}^{{}}$, where $\mathbf{K}$ is the constant symmetric nondegenerate
matrix 
\begin{equation}
K^{ij}=%
\begin{pmatrix}
0 & 0 & 0 & 0 & 0 & -2 \\ 
0 & 1 & -1 & -1 & -1 & 0 \\ 
0 & -1 & 1 & -1 & -1 & 0 \\ 
0 & -1 & -1 & 1 & -1 & 0 \\ 
0 & -1 & -1 & -1 & 1 & 0 \\ 
-2 & 0 & 0 & 0 & 0 & 0%
\end{pmatrix}%
,  \label{eq:14}
\end{equation}%
namely 
\begin{equation}
u_{y}^{i}=K^{ij}\partial _{x}^{{}}\frac{\delta \mathbf{H}_{7}}{\delta u^{j}},%
\text{ \ }u_{z}^{i}=K^{ij}\partial _{x}^{{}}\frac{\delta \mathbf{H}_{8}}{%
\delta u^{j}},  \label{eq:15}
\end{equation}%
where the Hamiltonian densities $h_{7}=a^{5}$ and $h_{8}=\frac{1}{2}a^{6}$,
while the momentum density $h_{6}=a^{3}$.

In the original coordinates $a^{k}(\mathbf{u})$ this Dubrovin--Novikov type
first order Hamiltonian operator becomes 
\begin{equation}
A_{1}^{ij}=%
\begin{pmatrix}
0 & 0 & 0 & -1 & 0 & 0 \\ 
0 & -1 & 0 & 0 & 0 & 0 \\ 
a^{1} & a^{2} & a^{3} & a^{4} & a^{5} & a^{6} \\ 
-1 & 0 & 0 & 0 & 0 & 0 \\ 
a^{2} & a^{4} & a^{5} & R & P & S \\ 
2a^{3} & 2a^{5} & 2a^{6} & 2P & 2S & 2Q%
\end{pmatrix}%
\partial _{x}^{{}}+\partial _{x}^{{}}%
\begin{pmatrix}
0 & 0 & a^{1} & -1 & a^{2} & 2a^{3} \\ 
0 & -1 & a^{2} & 0 & a^{4} & 2a^{5} \\ 
0 & 0 & a^{3} & 0 & a^{5} & 2a^{6} \\ 
-1 & 0 & a^{4} & 0 & R & 2P \\ 
0 & 0 & a^{5} & 0 & P & 2S \\ 
0 & 0 & a^{6} & 0 & S & 2Q%
\end{pmatrix}%
,  \label{eq:32}
\end{equation}%
where%
\begin{equation*}
P=\frac{a^{3}a^{4}+a^{6}}{a^{1}},\text{ \ }R=\frac{2a^{5}+a^{2}a^{4}}{a^{1}},%
\text{ \ }S=\frac{2a^{3}a^{5}-a^{2}a^{6}}{a^{1}},\text{ \ }%
Q=(a^{5})^{2}-a^{4}a^{6}+\frac{%
(a^{3})^{2}a^{4}+a^{3}a^{6}-2a^{2}a^{3}a^{5}+(a^{2})^{2}a^{6}}{a^{1}}.
\end{equation*}

\subsection{The Structure of the Second Hamiltonian Operator}

\label{sec:struct-second-hamilt}

Having in mind the fundamental examples of the case $n=N=3$
\cite{FGMN,KN1,KN2} we \emph{conjecture} that also our six-component systems
admit a Dubrovin--Novikov type third-order Hamiltonian operator, denoted by
$\hat{A}_{2}$.

Any bi-Hamiltonian hierarchy
\begin{equation}
u_{t^{k}}^{i}=A_{2}^{is}\frac{\delta \mathbf{H}_{k}}{\delta u^{s}}=A_{1}^{is}%
\frac{\delta \mathbf{H}_{k+1}}{\delta u^{s}},\text{ \ }i=1,...,n  \label{a}
\end{equation}%
is determined by two compatible Hamiltonian operators $\hat{A}_{1}$ and $
\hat{A}_{2}$, if their Schouten bracket vanishes: $[\hat{A}_{1},\hat{A}
_{2}]=0$. We recall that the Hamiltonian property of a differential operator
$\hat{A}$ is equivalent to its formal skew-adjointness $\hat{A}^{\ast }=-%
\hat{A}$ and $[\hat{A},\hat{A}]=0$, or the fact that $\hat{A}$ defines a
Poisson bracket on the space of conservation law densities:
\begin{displaymath}
  \{\mathbf{H}_k,\mathbf{H}_m\}_{\hat{A}} =
  \int \frac{\delta \mathbf{H}_k}{\delta u^i}A^{ij}\frac{\delta\mathbf{H}_m}
  {\delta u^j}dx
\end{displaymath}
According to Magri's Theorem (see \cite{magri}) all functionals
$\mathbf{H}_{p}=\int
h_{p}(\mathbf{u},\mathbf{u}_{x},\mathbf{u}_{xx},\mathbf{...})dx$ where $h_p$ is
a conservation law density commute with each other under either one of the
Poisson brackets defined by the operators
(i.e. $\{\mathbf{H}_{k},\mathbf{H}_{m}\}_{1}=0$ and independently
$\{\mathbf{H}_{k},\mathbf{H}_{m}\}_{2}=0$). In general, the reconstruction of
a bi-Hamiltonian nature of an integrable evolutionary system
  starting from the hierarchy of conservation law densities is a very
complicated problem even if both Hamiltonian operators are local. Now we assume
that a given evolutionary system
\begin{equation*}
u_{t}^{i}=U^{i}(\mathbf{u},\mathbf{u}_{x},\mathbf{u}_{xx},\mathbf{...})
\end{equation*}%
has a compatible pair of Dubrovin--Novikov type first and third order
Hamiltonian operators $\hat{A}_{1}$, $\hat{A}_{2}$, respectively. This means
that the above evolutionary system can be written in two different forms
(see (\ref{a})) 
\begin{equation}
u_{t}^{i}=A_{2}^{is}\frac{\delta \mathbf{H}_{1}}{\delta u^{s}}=A_{1}^{is}%
\frac{\delta \mathbf{H}_{2}}{\delta u^{s}},\text{ \ }i=1,...,n,  \label{e}
\end{equation}%
where (see detail in \cite{DN} and \cite{DN2}) 
\begin{equation}
A_{1}^{ij}=g_{1}^{ij}(\mathbf{u})\partial _{x}+b_{1k}^{ij}(\mathbf{u}%
)u_{x}^{k},  \label{z}
\end{equation}%
\begin{multline}
A_{2}^{ij}=g_{2}^{ij}(\mathbf{u})\partial _{x}^{3}+b_{2k}^{ij}(\mathbf{u}%
)u_{x}^{k}\partial _{x}^{2}+[c_{2k}^{ij}(\mathbf{u})u_{xx}^{k}+c_{2km}^{ij}(%
\mathbf{u})u_{x}^{k}u_{x}^{m}]\partial _{x}  \label{eq:17} \\
+d_{2k}^{ij}(\mathbf{u})u_{xxx}^{k}+d_{2km}^{ij}(\mathbf{u}%
)u_{xx}^{k}u_{x}^{m}+d_{2kmp}^{ij}(\mathbf{u})u_{x}^{k}u_{x}^{m}u_{x}^{p},
\end{multline}%
while the Hamiltonian functionals $\mathbf{H}_{1}=\int h_{1}(\mathbf{u},%
\mathbf{u}_{x},\mathbf{u}_{xx},\mathbf{...})dx$ and $\mathbf{H}_{2}=\int
h_{2}(\mathbf{u},\mathbf{u}_{x},\mathbf{u}_{xx},\mathbf{...})dx$ will be
specified in the next Subsection.

Let us recall some basic facts on the geometry of homogeneous first and
third-order Hamiltonian operators. We suppose that these Hamiltonian
operators have non-degenerate leading terms, i.e. $\det g_{1}^{ij}\neq 0$
and $\det g_{2}^{ij}\neq 0$. The Hamiltonian operators~(\ref{z}) and %
\eqref{eq:17} are form-invariant under point transformations of the
dependent variables, $\tilde{u}=\tilde{u}(\mathbf{u})$. In this case the
coefficients of both Hamiltonian operators~(\ref{z}) and (\ref{eq:17})
transform as differential-geometric objects. For instance, $g_{1}^{ij}$ and $%
g_{2}^{ij}$ transform as a $(2,0)$-tensors, so that their inverse $%
g_{ij}^{1} $ and $g_{ij}^{2}$ define pseudo-Riemannian metrics (note that
the second one is not flat in general), the expressions $%
g_{js}^{1}b_{1k}^{si}$, $-\frac{1}{3}g_{js}^{2}b_{2k}^{si}$, $-\frac{1}{3}%
g_{js}^{2}c_{2k}^{si}$, $-g_{js}^{2}d_{2k}^{si}$ transform as Christoffel
symbols of affine connections, etc. \cite{DN2}. B.A. Dubrovin and S.P.
Novikov proved that the $\hat{A}_{1}$ is Hamiltonian if and only if $\Gamma
_{1jk}^{i}=-g_{js}^{1}b_{1k}^{si}$ is the Levi-Civita connection of the
metric $g_{ij}^{1}$ and is flat \cite{DN}. It was conjectured by S.P.
Novikov that the last connection, $\Gamma _{2jk}^{i}=-g_{js}^{2}d_{2k}^{si}$
must be symmetric (with respect to low indices) and flat; this was confirmed
in \cite{GP91}, see also \cite{Doyle}. Therefore, there exists a coordinate
system $a^{k}(\mathbf{u})$ such that $\Gamma _{2jk}^{i}$ vanish. These
coordinates are determined up to affine transformations. We note here that $%
a^{i}$ are nothing but the conservation law densities of Casimirs of $\hat{A}%
_{2}$. In these coordinates $a^{k}$ the last three terms in (\ref{eq:17})
vanish, leading to the simplified expression \cite{GP97}, 
\begin{equation}
A_{2}^{ij}=\partial _{x}^{{}}\left( g^{ij}\partial
_{x}^{{}}+c_{k}^{ij}a_{x}^{k}\right) \partial _{x}^{{}}.  \label{casimir}
\end{equation}%
Here and below we omit the index $2$ in the notation of metric coefficients $%
g_{ik}^{2}(\mathbf{a})$ and connection coefficients $c_{2k}^{ij}(\mathbf{a})$.
In \cite{fpv} (using results from \cite{GP97}) it was proved that this
operator $\hat{A}_{2}$ is Hamiltonian (ie skew-adjoint and with
  vanishing Schouten bracket) if and only if the following system is
satisfied:
\begin{subequations}
\begin{align}
& c_{skm}=\frac{1}{3}(g_{sm,k}-g_{sk,m}),  \label{eq:19} \\
& g_{mk,s}+g_{ks,m}+g_{ms,k}=0,  \label{eq:20} \\
& c_{msk,l}=-g^{pq}c_{pml}c_{qsk}.  \label{eq:21}
\end{align}%
where $c_{ijk}=g_{iq}g_{jp}c_{k}^{pq}$. These conditions are invariant under
a class of reciprocal transformations that include affine transformations of
the flat coordinates of the last connection. Note that metrics fulfilling %
\eqref{eq:20} are Monge metrics of quadratic line complexes \cite%
{Dolgachev,fpv}.

Finding a Hamiltonian formulation that involves a third-order operator $\hat{%
A}_{2}$ for a hydrodynamic type system is not simple because the Hamiltonian
density will be non-local (see~\eqref{eq:6} and below the Hamiltonian
density $h_{1}$) in hydrodynamic variables $a^{k}$. However, the WDVV
hydrodynamic type systems are systems of conservation laws (see \eqref{eq:10}%
) 
\end{subequations}
\begin{equation}
a_{t}^{i}=(v^{i}(\mathbf{a}))_{x}.  \label{aa}
\end{equation}%
This means that after a potential substitution $a^{i}=b_{x}^{i}$ we obtain
the \textit{nonlinear} system $b_{t}^{i}=v^{i}(\mathbf{b}_{x})$ and the
Hamiltonian can be obtained by solving the following system of PDEs:%
\begin{equation}
v^{i}(\mathbf{b}_{x})=-(g^{ij}(\mathbf{b}_{x})\partial _{x}+c_{k}^{ij}(%
\mathbf{b}_{x})b_{xx}^{k})\frac{\delta \mathbf{H}}{\delta b^{j}}.
\label{eq:57}
\end{equation}%
In the field variables $b^{k}$ the Hamiltonian density
becomes\footnote{In general such a density depends polynomially on the
  independent variable $x$. However, in the cases considered in this paper the
  Hamiltonian densities do not depend explicitly on $x$.}
\textit{local}
(see detail in Subsection \ref{sec:hamiltonian}).

\subsection{Reconstructing the Second Hamiltonian Operator}

The metric coefficients $g_{ik}(\mathbf{a})$ completely determine
a Dubrovin--Novikov type third order Hamiltonian operator $\hat{A}_{2}$.
Indeed, the connection coefficients $c_{ijk}(\mathbf{a})$ are expressible
via metric coefficients (see \eqref{eq:19}), while the metric coefficients
must fulfill the Potemin system (here we preserved \eqref{eq:20} and
substituted \eqref{eq:19} into \eqref{eq:21})%
\begin{equation*}
g_{mk,s}+g_{ks,m}+g_{ms,k}=0,
\end{equation*}%
\begin{equation}
g_{mk,sl}-g_{ms,kl}=-\frac{1}{3}g^{pq}(g_{pl,m}-g_{pm,l})(g_{qk,s}-g_{qs,k}).
\label{q}
\end{equation}%
Thus if some candidate to be metric coefficients $g_{ik}(\mathbf{u})$
are found in an arbitrary coordinate system $u^{k}$, one must look for point
transformations $a^{k}(\mathbf{u})$ such that the metric coefficients
$g_{ik}(\mathbf{a} )=g_{ms}(\mathbf{u})\frac{\partial u^{m}}{\partial
  a^{i}}\frac{\partial u^{s} }{\partial a^{k}}$ will satisfy the Potemin
system. Theoretically this is a very complicated task.  However, for
particular cases, for instance, for non-diagonalizable hydrodynamic type
systems this is an algorithmically solvable problem.

\textbf{Our main observation} is: if an integrable hierarchy of evolutionary
equations (\ref{a}) contains a commuting flow, which is a hydrodynamic type
system (see, for instance, \eqref{eq:10}, (\ref{aa})), then:

1. its conservation law densities are quasi-homogeneous polynomials (i.e.  they
are homogeneous polynomials with respect to any derivatives of
\textquotedblleft $x$\textquotedblright , but coefficients could
depend on the field variables $u^{k}$ in an arbitrary way). For instance, the
first two higher conservation law densities of them have the form\footnote{%
  any conservation law density is determined up to a total $x$-derivative}
\begin{equation}
h_{1}=a_{sm}(\mathbf{u})u_{x}^{s}u_{x}^{m},\text{ \ }h_{2}=a_{ms}^{(1)}(%
\mathbf{u})u_{xx}^{m}u_{xx}^{s}+a_{lms}^{(2)}(\mathbf{u}%
)u_{xx}^{l}u_{x}^{m}u_{x}^{s}+a_{lsmp}^{(3)}(\mathbf{u}%
)u_{x}^{l}u_{x}^{s}u_{x}^{m}u_{x}^{p}.  \label{h}
\end{equation}

2. the commuting flow (\ref{e}) becomes\footnote{%
Here by \textquotedblleft
l.o.t.\textquotedblright\ we mean \textquotedblleft lower order
terms\textquotedblright .}%
\begin{equation}
u_{t}^{i}=(g_{2}^{ip}\partial _{x}^{3}+\text{l.o.t.})(-2a_{pm}u_{xx}^{m}+%
\text{l.o.t.})=(g_{1}^{ip}\partial _{x}+\text{l.o.t.}%
)(2a_{pm}^{(1)}u_{xxxx}^{m}+\text{l.o.t.}).  \label{recurs}
\end{equation}
Then we obtain the relationship in highest order terms (the
  coefficient of $u_{xxxxx}^{m}$)
\begin{equation}
g_{2}^{ip}a_{pm}=-g_{1}^{ip}a_{pm}^{(1)}.  \label{eq:16}
\end{equation}%
Thus if the metric coefficients $g_{1}^{ik}(\mathbf{u})$ are known and $\det
a_{jm}\neq 0$, then the metric coefficients of the second operator
  can be found by
\begin{equation}
g_{2}^{ij}=-g_{1}^{ip}a_{pm}^{(1)}c^{mj},  \label{metr}
\end{equation}%
where $a_{im}c^{mj}=\delta _{i}^{j}$ and $c^{ip}a_{pj}=\delta _{j}^{i}$.

Indeed the commuting hydrodynamic type system has a local Hamiltonian
structure (see (\ref{z}))%
\begin{equation*}
u_{t^{-1}}^{i}=A_{1}^{is}\frac{\delta \mathbf{H}_{0}}{\delta u^{s}},
\end{equation*}%
where the Hamiltonian density $h_{0}(\mathbf{u})$ depends on field variables 
$u^{k}$ only. Then the next commuting flow%
\begin{equation*}
u_{t^{0}}^{i}=A_{2}^{is}\frac{\delta \mathbf{H}_{0}}{\delta u^{s}}=A_{1}^{is}%
\frac{\delta \mathbf{H}_{1}}{\delta u^{s}}
\end{equation*}%
is an evolutionary system of \textit{third} order (see
\eqref{eq:17}). In view
of the homogeneity of the operators $\hat{A}_1$ and $\hat{A}_2$,  this
means that $h_{1}$ can be chosen in above (left) form (\ref{h})
  up to total $x$-derivatives:
\begin{equation*}
\tilde{h}_{1}=\tilde{a}_{sm}(\mathbf{u})u_{x}^{s}u_{x}^{m}+b_{s}(\mathbf{u}%
)u_{xx}^{s}=[\tilde{a}_{sm}(\mathbf{u})-b_{s,m}(\mathbf{u}%
)]u_{x}^{s}u_{x}^{m}+(b_{s}(\mathbf{u})u_{x}^{s})_{x}.
\end{equation*}%
Then next commuting flow (\ref{e}) is an evolutionary system of \textit{fifth%
} order. This means that $h_{2}$ can be chosen in above (right) form (\ref{h}%
).

Thus we constructed a link between metric coefficients $g_{1}^{ip},g_{2}^{ip}$
and coefficients $a_{sm}(\mathbf{u}),a_{ms}^{(1)}(\mathbf{u})$
in (\ref{h}) only. This means: if one knows a metric $g_{1}^{ip}$ and two
conservation law densities $h_{1},h_{2}$, then metric coefficients $g_{2}^{ip}$
can be found from (\ref%
{metr}).

% We stress that no useful relationship can be obtained if one of the
% conservation law densities gives rise to a trivial identity~\eqref{eq:16}.
% This is the case, for example, if one of the conservation law densities is
% trivial, \emph{i.e.} if it is a total derivative with respect to
% \textquotedblleft $x$\textquotedblright\ of some expression which depends on
% field variables $u^{k}$ and their finite number of higher derivatives.

In general, any integrable hydrodynamic-type system of (nonlinear) evolutionary
PDEs possesses infinitely many local conservation laws of arbitrary order with
respect to higher derivatives of the field variables $u^{k}$ (of the
independent variable \textquotedblleft $x$\textquotedblright\ only)
\begin{equation*}
(h(\mathbf{u},\mathbf{u}_{x},\mathbf{u}_{xx},\mathbf{...)})_{t}=(f(\mathbf{u}%
,\mathbf{u}_{x},\mathbf{u}_{xx},\mathbf{...}))_{x}.
\end{equation*}%
However in practice conservation laws like in \eqref{h} cannot be easily
found even in three-component case; in the six component case the direct
search of such conservation law densities is probably impossible by any
computer algebra system on existing workstations. Nevertheless, this problem
is effectively solvable if a Lax pair is known. In such a case the
complexity of computation is determined by the complexity of a Taylor
expansion (see details below).

In general any integrable system of nonlinear evolutionary PDEs has at least a
finite number of linearly independent hydrodynamic conservation law densities
$h(\mathbf{u})$. In the non-diagonalizable case (like the WDVV systems), if the
system is endowed by a first-order Dubrovin-Novikov homogeneous Hamiltonian
operator, the dimension of the space of hydrodynamic integrals is $n+2$, where
$n$ is the number of components \cite{fs}. Flat coordinates for the Hamiltonian
operator can be selected among the hydrodynamic integrals (see details in
\cite{MaksTsarEgor}). For instance, if $N=3$, then WDVV associativity equation
reduces to a three-component non-diagonalizable hydrodynamic type system which
has just \textit{five} hydrodynamic conservation law densities (see details in
\cite{FGMN}); if $N=4$, then the WDVV associativity equations reduce to a
compatible pair of six-component non-diagonalizable hydrodynamic type (see
\eqref{eq:10} and \eqref{eq:15}) which have just \textit{nine} hydrodynamic
conservation law densities, i.e. flat coordinates of the first Hamiltonian
structure $u^{0}=a^{1},u^{1},u^{2},u^{3},u^{4},u^{5}=a^{4}$, a momentum density
$a^{3}$ quadratic in field variables $u^{k}$ and two Hamiltonian densities
$a^{5},a^{6}$ rational expressions with respect to these flat coordinates. In
these flat coordinates $g_{1}^{ij}=K^{ij}$ is a constant symmetric
non-degenerate matrix. Thus\footnote{%
  here again and everywhere below we omit the index $2$ of the metric
  $g_{2}^{ik}$.} (see (\ref{metr}))
\begin{equation}
g^{ij}(\mathbf{u})=-K^{ip}a_{pm}^{(1)}c^{mj}.  \label{k}
\end{equation}%
Once the metric coefficients $g^{ij}(\mathbf{u})$ of the contravariant metric
are found we want to prove that the covariant metric
$g_{ik}(\mathbf{a})=g_{ms}(\mathbf{u})\frac{\partial u^{m}}{\partial
  a^{i}}\frac{\partial u^{s}}{\partial a^{k}}$ is a Monge metric. To do that
one should first find Casimir densities $a^{k}$ of a Dubrovin--Novikov type
third order Hamiltonian operator $\hat{A}_{2}$. In an arbitrary coordinate
system $u^{k}$ the metric coefficients $g^{ij}(\mathbf{u})$ cannot completely
determine such a Hamiltonian operator $\hat{A}_{2}$ (see \eqref{eq:17}), even
if all other coefficients are connected with each other via skew-symmetry and
Jacobi identity conditions. Thus finding the Casimir densities $a^{k}(\mathbf{u
})$ is really an important problem. Theoretically they can be found easily,
because they belong to a finite number of hydrodynamic conservation law
densities, which we already discussed above.  However, in our case, we have
natural candidates to be Casimir densities.  Indeed, below we prove that
Casimir densities $a^{k}(\mathbf{u})$ of a Dubrovin--Novikov type third order
Hamiltonian operator are precisely coordinates $a^{k}$ from \eqref{eq:10},
\eqref{eq:11}, whose relationship with flat coordinates $u^{s}$ of a
Dubrovin--Novikov type first order Hamiltonian operator is given according to
the Vi\`{e}te formulae by \eqref{eq:20}.

In our case non-diagonalizable hydrodynamic type systems \eqref{eq:10}, %
\eqref{eq:15} admit four infinite sequences of homogeneous conservation law
densities. The sequences can be deduced from the Lax pair by a standard
technique known in integrable systems as follows. The members of such
sequences of degree $2$ (\emph{i.e.}, whose densities are quadratic in
velocities) fulfill a nondegeneracy hypothesis and are enough to reconstruct
the leading term of the Hamiltonian operator $\hat{A}_{2}$.

By eliminating $\psi _{1}$, $\psi _{2}$, $\psi _{3}$ from \eqref{eq:11} we
obtain the single linear PDE 
\begin{multline*}
\lambda ^{2}\frac{a^{3}}{a^{1}}\psi _{xx}+\lambda ^{3}\left( a^{5}-\frac{%
a^{2}a^{3}}{a^{1}}\right) \psi _{x}+\lambda ^{4}\left( a^{6}-\frac{%
(a^{3})^{2}}{a^{1}}\right) \psi \\
=\left( \frac{1}{a^{1}}\psi _{xx}-\lambda \frac{a^{2}}{a^{1}}\psi
_{x}-\lambda ^{2}\frac{a^{3}}{a^{1}}\psi \right) _{xx}+\left[ \lambda
^{3}\left( \frac{a^{2}a^{3}}{a^{1}}-a^{5}\right) \psi +\lambda ^{2}\left( 
\frac{(a^{2})^{2}}{a^{1}}-a^{4}\right) \psi _{x}-\lambda \frac{a^{2}}{a^{1}}%
\psi _{xx}\right] _{x}.
\end{multline*}%
The substitution $\psi =\exp \int rdx$ yields a nonlinear ordinary
differential equation on the function $r$ and its first, second and third
order derivatives. This function $r$ plays the role of a generating function
of conservation law densities with respect to the parameter $\lambda $ for
both systems \eqref{eq:10}. The expansion of $r$ at infinity (i.e. $\lambda
\rightarrow \infty $)%
\begin{equation*}
r=\lambda h_{-1}+h_{0}+\frac{h_{1}}{\lambda }+\frac{h_{2}}{\lambda ^{2}}+...,
\end{equation*}%
in the above equation leads to a sequence of differential relationships
between the coefficients $h_{-1}$, $h_{0}$, $h_{1}$,\dots\ The leading term
(the coefficient of $\lambda ^{3}$) coincides with the characteristic
equation of the eigenvalues of the matrix $\mathbf{A}$ (see \eqref{eq:12}).
Thus the expansion $r$ with respect to the parameter $\lambda $ has four
branches of conservation law densities, where we identify $h_{-1}$ to each
of four roots $u^{k}$ of characteristic polynomial \eqref{eq:12},
correspondingly. So, changing coordinates $a^{k}$ to flat coordinates $u^{m}$
(see \eqref{eq:13}, \eqref{eq:14}, \eqref{eq:15}, \eqref{eq:32}) in these
expansions we have four branches of conservation law densities, \emph{i.e.} 
\begin{equation*}
r_{(k)}=\lambda u^{k}+h_{0k}[\mathbf{u}]+\frac{h_{1k}[\mathbf{u}]}{\lambda }+%
\frac{h_{2k}[\mathbf{u}]}{\lambda ^{2}}+...,\text{ \ }k=1,2,3,4,
\end{equation*}%
whose coefficients have more or less compact form (in comparison with
expressions $h_{ik}[\mathbf{a}]$ in original coordinates $a^{k}(\mathbf{u})$%
) and depend on field variables $u^{k}$ as well as their higher derivatives
with respect to the independent variable \textquotedblleft $x$%
\textquotedblright .

The expressions of such conservation law densities are quasihomogeneous
polynomials of degrees $\deg h_{ik}=i+1$ with respect to the grading $\deg
u=0$, $\deg \partial _{x}^{{}}=1$, and their coefficients are expressible via
rational functions of these field variables $u^{k}$. We computed (using Reduce
\cite{cde}) all expressions of $h_{ik}$, for $k=1$, $2$, $3$, $4$ and $i=0$,
$1$, $2$, $3$, the non-trivial ones are of the form (cf. (\ref{h}))
\begin{align}
h_{1k}& =-\frac{1}{2}G_{ksm}(\mathbf{u})u_{x}^{s}u_{x}^{m},  \label{eq:29} \\
h_{3k}& =Q_{kms}^{(1)}(\mathbf{u})u_{xx}^{m}u_{xx}^{s}+Q_{klms}^{(2)}(%
\mathbf{u})u_{xx}^{l}u_{x}^{m}u_{x}^{s}+Q_{klsmp}^{(3)}(\mathbf{u}%
)u_{x}^{l}u_{x}^{s}u_{x}^{m}u_{x}^{p}  \label{eq:100}
\end{align}%
Finding the expressions is not very heavy, but the results need
simplification. Indeed, simplification of the rational expressions of the
coefficients of derivatives consumes the greatest part of computing resources
(18GB of RAM and 6 hours of CPU time). Even after the simplification, the
expressions of the coefficients are huge and it is not worth to write them
down. This is not a problem as the task is to find the Hamiltonian structure
which is checkable by pen and paper (after it has been found!), as one can see
below.

We recall that in the case $N=3$ the above procedure leads to $3$
conservation law densities $h_{1k}$ ($k=1,2,3$), and that their sum is zero 
\cite{FGMN}. We have a similar result in our case $N=4$.

\begin{proposition}
The conservation law densities $h_{11}(\mathbf{u})$, $h_{12}(\mathbf{u})$, $%
h_{13}(\mathbf{u})$ are linearly independent, and we have $
\sum_{k=1}^{4}h_{1k}(\mathbf{u})=0$.
\end{proposition}
\begin{proof}
The proof is trivial using a computer algebra system.
\end{proof}

\begin{remark}
  The conservation law densities $h_{0k}$ and $h_{2k}$ are trivial for
  three-component WDVV associativity equations in \cite{FGMN}, and $h_{0k}$ is
  trivial in our six-component hydrodynamic type systems (the
    proof of triviality is done by CDE \cite{cde}). For instance, in the
  three-component WDVV case we have
\begin{equation*}
h_{01}=-\frac{1}{2}\partial _{x}\ln (u^{1}-u^{2})(u^{1}-u^{3}).
\end{equation*}%
However, we have no general proof that this holds for all even-indexed
conservation law densities in the above expansion or moreover for all higher
WDVV associativity equations.
\end{remark}

Now, we show that the above sequence of conservation law densities %
\eqref{eq:29}, \eqref{eq:100} enables us to reconstruct the metric
coefficients $g^{ik}(\mathbf{a})$ of the leading term $\partial _{x}^{3}$ in
the second Hamiltonian operator $\hat{A}_{2}$. Below we formulate the
Theorem, whose validity is general and is not limited to our particular
hydrodynamic type systems of PDEs. However, this Theorem is based on the
interplay between two coordinate systems, i.e. flat coordinates $u^{k}(%
\mathbf{a})$ of a first local homogeneous Hamiltonian structure of first
order and Casimirs $a^{k}(\mathbf{u})$ of a second local homogeneous
Hamiltonian structure of third order (see again \eqref{eq:13}, \eqref{eq:14},
\eqref{eq:15}, \eqref{eq:32}). In our case, for any fixed $i=1,2,3,4$ the
matrices $G_{ijk}(\mathbf{u})$ are degenerate. However, we can introduce the
matrix $\mathbf{\tilde{G}}$ whose coefficients are $\xi ^{m}G_{mpq}$, where $%
\xi ^{k}$ are arbitrary constants. Without loss of generality we may choose
any constants $\xi ^{k}$ such that the matrix $\mathbf{\tilde{G}}$ is
non-degenerate; then we can introduce the inverse matrix $\mathbf{C}$ whose
coefficients we denote $C^{km}$.

\begin{theorem}
\label{th:reconstruction} Let $\hat{A}_{1}$ and $\hat{A}_{2}$ be two
Dubrovin--Novikov type Hamiltonian operators of first and of third order,
respectively, for hydrodynamic type systems \eqref{eq:5}, \eqref{eq:10},
written in flat coordinates $u^{k}$, i.e. $\hat{A}_{1}=\mathbf{K}%
\partial_{x} $.

Let $\mathbf{H}_{1k}=\int h_{1k}dx$ and $\mathbf{H}_{2k}=\int h_{3k}dx$ be
two sets of homogeneous conservation law functionals of degrees $2$ and $4$,
respectively, for any $k$ in a given range;

Then the metric coefficients $g^{ik}(\mathbf{u})$ of a Dubrovin--Novikov
type third order Hamiltonian operator $\hat{A}_{2}$ are uniquely determined
by the formula 
\begin{equation}
g^{ij}(\mathbf{u})=2\xi ^{m}K^{ip}Q_{mpq}^{(1)}C^{qj}.  \label{eq:30}
\end{equation}
\end{theorem}

\begin{proof}
Let us consider the recurrence relation on $\mathbf{H}_{1k}$ and $\mathbf{H}%
_{2k}$ 
\begin{equation}
A_{1}^{im}\frac{\delta \mathbf{H}_{2k}}{\delta u^{m}}=A_{2}^{im}\frac{\delta 
\mathbf{H}_{1k}}{\delta u^{m}}  \label{eq:31}
\end{equation}%
(see equations \eqref{eq:29}, \eqref{eq:100}). Then the variational
derivatives can be rewritten as 
\begin{align*}
\frac{\delta \mathbf{H}_{1k}}{\delta u^{j}}& =G_{kjm}u_{xx}^{m}+\text{lower
order terms}, \\
& \frac{\delta \mathbf{H}_{2k}}{\delta u^{j}}=2Q_{kjm}^{(1)}u_{xxxx}^{m}+%
\text{lower order terms}.
\end{align*}%
Then we have%
\begin{align*}
A_{1}^{ij}\frac{\delta \mathbf{H}_{2k}}{\delta u^{j}}& =K^{ij}\partial
_{x}^{{}}(2Q_{kjm}^{(1)}u_{xxxx}^{m}+\text{lower order terms}) \\
& =2K^{ij}Q_{kjm}^{(1)}u_{xxxxx}^{m}+\text{lower order terms}, \\
A_{2}^{ij}\frac{\delta \mathbf{H}_{1k}}{\delta u^{j}}&
=A_{2}^{ij}(G_{kjm}u_{xx}^{m}+\text{lower order terms}) \\
& =g^{ij}G_{kjm}u_{xxxxx}^{m}+\text{lower order terms}.
\end{align*}%
By equating the coefficients of highest order derivatives we have $%
2K^{im}Q_{kms}^{(1)}=g^{im}G_{kms}$ for $k=1$, $2$, $3$, $4$. By taking the
linear combination with coefficients $\xi ^{k}$ of the above identities we
can invert the matrix $\xi ^{k}G_{kpq}$ on the right-hand side, which yields
the result \eqref{eq:30} (cf. (\ref{k})). The Theorem is proved.
\end{proof}

The sequence of homogeneous conservation law densities for our hydrodynamic
type systems \eqref{eq:10} fulfills the nondegeneracy condition which is
necessary in the above Theorem.

\begin{proposition}
\label{pr:leadingterm} The following linear combination of the four
conservation law densities of degree $2$: 
\begin{equation*}
C_{ij}=G_{1ij}+\frac{1}{3}G_{2ij}+\frac{1}{2}G_{4ij}
\end{equation*}%
fulfills $\det (C_{ij})\neq 0$.
\end{proposition}
\begin{proof}
  The proof is trivial using a computer algebra system.
\end{proof}

We observe that the proof would be very lengthy by pen and paper: the core
computation would be a determinant of a $6\times 6$ matrix whose entries are
rational expressions of degree up to $6$ in both the numerator and the
denominator.

\begin{remark}
One might be tempted to use the recurrence relation~\eqref{eq:31} by a
simpler conservation law density. Indeed, the first nontrivial densities are
the hydrodynamic type densities. However, this choice on the right-hand side
of~\eqref{eq:31} leads to a trivial identity or to an identity which is not
related with the leading coefficient $g^{ij}$. So, the first `useful'
conservation law densities to our purposes are exactly those quadratic in
first derivatives.
\end{remark}

\begin{corollary}
\label{co:leadingterm} If the hydrodynamic type systems \eqref{eq:10} admit
a second Dubrovin--Novikov type third-order Hamiltonian operator $\hat{A}%
_{2} $, then the metric coefficients $g^{ik}(\mathbf{u})$ (i.e. the
coefficients of the leading term $g^{ik}\partial _{x}^{3}$) have the form
which is specified by formula \eqref{eq:30} with $G_{kjm}$ and $%
Q_{kjm}^{(1)} $ found in \eqref{eq:29} and \eqref{eq:100}, and here we
chosen $\xi ^{1}=1$, $\xi ^{2}=1/3$, $\xi ^{3}=0$, $\xi ^{4}=1/2$.
\end{corollary}

The expression of $g^{ij}$ can be found using results from
Theorem~\ref{th:reconstruction}. The explicit computation can be carried out by
Reduce \cite{cde}. Here most computational resources are consumed by the
simplification of the final rational expression of $g^{ij}$. Even after
simplification the expression in coordinates $u^{k}$ (flat coordinates of the
Dubrovin--Novikov type first order Hamiltonian structure) is huge and it is not
worth writing it here.

\section{Integrability of the Potemin System}

\label{sec:potemin}

In this Section we consider the integrability of Potemin system \eqref{eq:19}%
, \eqref{eq:20}, \eqref{eq:21} following from skew-symmetry and Jacobi
identity of homogeneous third order Hamiltonian operators (\ref{casimir}).

Let introduce the linear expressions 
\begin{equation}
\psi _{k}^{\gamma }=\psi _{km}^{\gamma }a^{m}+\omega _{k}^{\gamma },
\label{omega}
\end{equation}%
where $\psi _{km}^{\gamma }$ and $\omega _{k}^{\gamma }$ are such constants
that the matrix $\mathbf{\psi }$ is nondegenerate and%
\begin{equation}
\psi _{km}^{\gamma }=-\psi _{mk}^{\gamma },  \label{c}
\end{equation}%
\begin{equation}
\overset{n}{\sum_{\gamma =1}}(\psi _{is}^{\gamma }\psi _{jk}^{\gamma }+\psi
_{js}^{\gamma }\psi _{ki}^{\gamma }+\psi _{ks}^{\gamma }\psi _{ij}^{\gamma
})=0,\text{ \ }\overset{n}{\sum_{\gamma =1}}(\omega _{i}^{\gamma }\psi
_{jk}^{\gamma }+\omega _{j}^{\gamma }\psi _{ki}^{\gamma }+\omega
_{k}^{\gamma }\psi _{ij}^{\gamma })=0.  \label{b}
\end{equation}

\begin{theorem}
\label{th:zzz} [GVPotemin 2001] The metric coefficients $g_{ik}(\mathbf{a})$
of Potemin system \eqref{eq:19}, \eqref{eq:20}, \eqref{eq:21} can be
presented in the form%
\begin{equation}
g_{mk}=\overset{n}{\sum_{\gamma =1}}\psi _{m}^{\gamma }\psi _{k}^{\gamma }.
\label{tri}
\end{equation}
\end{theorem}

\begin{proof}
Taking into account (\ref{c}) and (\ref{tri}), the first equation %
\eqref{eq:19} of the Potemin system leads to 
\begin{equation}
c_{ijk}=-\overset{n}{\sum_{\gamma =1}}\psi _{i}^{\gamma }\psi _{jk}^{\gamma }
\label{pjat}
\end{equation}%
while the second equation \eqref{eq:20} of the Potemin system implies a set
of constraints 
\begin{equation}
\overset{n}{\sum_{\gamma =1}}(\psi _{i}^{\gamma }\psi _{jk}^{\gamma }+\psi
_{j}^{\gamma }\psi _{ki}^{\gamma }+\psi _{k}^{\gamma }\psi _{ij}^{\gamma
})=0,  \label{cycle}
\end{equation}%
which under the substitution $\psi _{k}^{\gamma }=\psi _{km}^{\gamma
}a^{m}+\omega _{k}^{\gamma }$ yields (\ref{b}). Then introducing the inverse
metric 
\begin{equation}
g^{mk}=\overset{n}{\sum_{\gamma =1}}\psi _{\gamma }^{m}\psi _{\gamma }^{k}
\label{four}
\end{equation}%
such that $\psi _{\gamma }^{i}\psi _{k}^{\gamma }=\delta _{k}^{i}$ and $\psi
_{\gamma }^{m}\psi _{m}^{\beta }=\delta _{\gamma }^{\beta }$, substitution
of (\ref{tri}) and (\ref{pjat}) into the third equation \eqref{eq:21} of the
Potemin system yields the identity.
\end{proof}

\begin{corollary}
\label{co:yyy} In the general ($n$ component) case constraints (\ref{b}) can
be resolved. For instance, if $n=3$, then these constraints reduce to a
single equation, i.e. to a sole equation in r.h.s. of (\ref{b}); if $n=4$,
then the constraints reduce to a system of five equations, which contains a
sole equation from l.h.s. of (\ref{b}) and four other equations from r.h.s.
of (\ref{b}).
\end{corollary}

\begin{corollary}
\label{co:xxx} 
\begin{equation*}
\det g_{ik}=(\det \psi _{m}^{\gamma })^{2}.
\end{equation*}
\end{corollary}

\subsection{Decomposition of the Monge metric}

In this Subsection we present an effective approach for the metric
decomposition (\ref{tri}). The linear system of PDEs 
\begin{equation}
\psi _{j,k}=\frac{1}{3}\psi _{p}g^{pq}(g_{qj,k}-g_{qk,j})  \label{d}
\end{equation}
on $n$ functions $\psi _{k}(\mathbf{a})$ can be interpreted as $n$ commuting
linear systems of ODEs for each fixed index $k$. Thus this linear system
possesses a general solution parametrized by $n$ arbitrary constants only.

It is easy to see that $(\psi _{j,k})_{m}=(\psi _{j,m})_{k}=0$. This means
that $\psi _{k}$ are \textit{linear} functions with respect to field
variables $a^{s}$.

Under a geometric viewpoint, the system \eqref{d} is equivalent to $\nabla
\psi =0$, where $\nabla $ is the linear connection determined by $c_{ijk}$.
Since the nonlinear condition of the Potemin system is just the requirement
of flatness of $\nabla $, the system \eqref{d} is just the equation of
parallel vectors of $\nabla $. Such an equation always admits $n$
independent solutions which, in our case, are linear, thus making their
search particularly simple.

Let us choose $n$ particular solutions $\psi _{k}^{\gamma }$ of this system%
\begin{equation*}
\psi _{j,k}^{\gamma }=\frac{1}{3}\psi _{p}^{\gamma }g^{pq}(g_{qj,k}-g_{qk,j})
\end{equation*}%
such that $\det \psi _{p}^{\gamma }\neq 0$. Then taking into account the
constraints (\ref{cycle}), the metric coefficients can be decomposed in the
form (\ref{tri}), (\ref{four}). The antisymmetric condition (\ref{c}) can be
easily obtained from the above linear system by permutation of the indices $%
j,k$. Since $\psi _{k}^{\gamma }$ are linear functions of $a^{k}$, then one
can identify $\psi _{j,k}^{\gamma }=\psi _{jk}^{\gamma }$ which are
skewsymmetric constants with respect to lower indices $j,k$.

\begin{remark}
For better convenience the metric decomposition formulas (\ref{tri}), (\ref%
{four}) can be slightly modified in the following way. For any Monge metric
satisfying the Potemin system \eqref{eq:19}, \eqref{eq:20}, \eqref{eq:21}
one can choose any new linear combination of elementary solutions $\psi
_{p}^{\gamma }$, because (\ref{d}) is a linear system. Thus once some $n$
particular solutions determine a non-degenerate matrix $\mathbf{\psi }$, one
can introduce the constant non-degenerate matrix $\mathbf{\phi }$ such that
(cf. (\ref{tri}))%
\begin{equation}
g_{ij}=\phi _{\beta \gamma }\psi _{i}^{\beta }\psi _{j}^{\gamma },
\label{gen}
\end{equation}%
where $\phi _{\beta \gamma }$ are elements of the matrix $\mathbf{\phi }$
and $\psi _{p}^{\gamma }$ are elements of the matrix $\mathbf{\psi }$.
Indeed, for any symmetric constant matrix $\mathbf{\phi}$ we have $\mathbf{%
\phi =J\Lambda J}^{T}$, where $\mathbf{\Lambda }$ is a diagonal matrix and $%
\mathbf{J}$ is an appropriate constant matrix. Introducing the new set of
particular solutions $\mathbf{\tilde{\psi}=J\psi }$, the above formula $%
g_{ij}=\phi _{\beta \gamma }\psi _{i}^{\beta }\psi _{j}^{\gamma }$ reduces
to the form $g_{ij}=\lambda _{\beta }\delta _{\beta \gamma }\tilde{\psi}%
_{i}^{\beta }\tilde{\psi} _{j}^{\gamma }$, where $\lambda _{\beta }$ are
diagonal elements of the matrix $\mathbf{\Lambda }$. Finally, by
appropriately scaling $\tilde{\psi} _{i}^{\beta }$, one can obtain again the
original formula (\ref{tri}).

In this extended construction the connection coefficients (\ref{pjat})
become $c_{ijk}=-\phi _{\beta \gamma }\psi _{i}^{\beta }\psi _{jk}^{\gamma }$%
, the skew-symmetry condition (\ref{c}) is the same, but (\ref{b}) takes the
form%
\begin{equation}
\phi _{\beta \gamma }(\psi _{is}^{\beta }\psi _{jk}^{\gamma }+\psi
_{js}^{\beta }\psi _{ki}^{\gamma }+\psi _{ks}^{\beta }\psi _{ij}^{\gamma
})=0,\text{ \ }\phi _{\beta \gamma }(\omega _{i}^{\beta }\psi _{jk}^{\gamma
}+\omega _{j}^{\beta }\psi _{ki}^{\gamma }+\omega _{k}^{\beta }\psi
_{ij}^{\gamma })=0.  \label{seven}
\end{equation}%
Since the matrices $\mathbf{\phi }$ and $\mathbf{\psi }$ are non-degenerate,
the inverse metric (cf. (\ref{four}))%
\begin{equation}
g^{ij}=\phi ^{\beta \gamma }\psi _{\beta }^{i}\psi _{\gamma }^{j}
\label{chet}
\end{equation}%
can be easily reconstructed.
\end{remark}

\subsection{Factorised Third Order Homogeneous Hamiltonian Structure}

\label{sec:fact-third-order}

Once a metric decomposition is found, the corresponding Dubrovin--Novikov
type differential-geometric third order Poisson bracket (see (\ref{casimir}%
)) 
\begin{equation*}
\{a^{i}(x),a^{j}(x^{\prime })\}_{2}=\partial _{x}(g^{ij}\partial
_{x}+c_{k}^{ij}u_{x}^{k})\delta ^{\prime }(x-x^{\prime })
\end{equation*}%
also can be written in the factorised form%
\begin{equation*}
\{a^{i}(x),a^{j}(x^{\prime })\}_{2}=\phi ^{\beta \gamma }\partial _{x}\psi
_{\beta }^{i}\partial _{x}\psi _{\gamma }^{j}\delta ^{\prime }(x-x^{\prime
}).
\end{equation*}%
This means that any evolutionary system equipped by a Dubrovin--Novikov type
third order Hamiltonian structure can be written in the form%
\begin{equation*}
a_{t}^{i}=\partial _{x}(g^{is}\partial _{x}+c_{k}^{is}a_{x}^{k})\partial _{x}%
\frac{\delta \mathbf{H}}{\delta a^{s}}=\phi ^{\beta \gamma }\partial
_{x}\psi _{\beta }^{i}\partial _{x}\psi _{\gamma }^{s}\partial _{x}\frac{%
\delta \mathbf{H}}{\delta a^{s}}.
\end{equation*}%
Thus a Dubrovin--Novikov type third order Hamiltonian operator (see (\ref%
{gen}))%
\begin{equation}
A_{2}^{ij}=\phi ^{\beta \gamma }\partial _{x}\psi _{\beta }^{i}\partial
_{x}\psi _{\gamma }^{j}\partial _{x}  \label{s}
\end{equation}%
can be found together with metric decomposition formulae (\ref{gen}), (\ref%
{chet}) in our six-component case (see next Section).

Below we provide explicit expressions for \textit{nonlocal} Casimirs and the
momentum density; then we reconstruct Hamiltonians for non-diagonalizable
hydrodynamic type systems equipped by a Dubrovin--Novikov type third order
Hamiltonian operator (\ref%
{s}).

\subsection{Nonlocal Casimirs}

\label{sec:casimirs}

Under the potential substitution $a^{i}=b_{x}^{i}$ the evolutionary
conservative system%
\begin{equation}
a_{t}^{i}=(V^{i}(\mathbf{a},\mathbf{a}_{x},\mathbf{a}_{xx},\mathbf{...}))_{x}
\label{g}
\end{equation}%
takes the form%
\begin{equation}
b_{t}^{i}=V^{i}(\mathbf{b}_{x},\mathbf{b}_{xx},\mathbf{...}).  \label{f}
\end{equation}%
Correspondingly, if this evolutionary system has a Dubrovin--Novikov type
third order Hamiltonian structure%
\begin{equation*}
a_{t}^{i}=\partial _{x}(g^{is}\partial _{x}+c_{k}^{is}a_{x}^{k})\partial _{x}%
\frac{\delta \mathbf{H}}{\delta a^{s}}=\phi ^{\beta \gamma }\partial
_{x}\psi _{\beta }^{i}\partial _{x}\psi _{\gamma }^{s}\partial _{x}\frac{%
\delta \mathbf{H}}{\delta a^{s}},
\end{equation*}%
then in potential variables $b^{i}$ we see that evolutionary system (\ref{f}%
) has a local Hamiltonian structure of \textit{first} order%
\begin{equation}
b_{t}^{i}=-(g^{is}\partial _{x}+c_{k}^{is}b_{xx}^{k})\frac{\delta \mathbf{H}%
}{\delta b^{s}}=-\phi ^{\beta \gamma }\psi _{\beta }^{i}\partial _{x}\psi
_{\gamma }^{s}\frac{\delta \mathbf{H}}{\delta b^{s}},  \label{bi}
\end{equation}%
but this is \textbf{not} a Dubrovin--Novikov type \textit{first} order
Hamiltonian structure, because its coefficients $g^{is}(\mathbf{b}_{x})$, $%
c_{k}^{is}(\mathbf{b}_{x})$, $\psi _{\beta }^{i}(\mathbf{b}_{x})$ have no
geometrical interpretation (i.e. they cannot change under arbitrary point
transformations $\tilde{b}^{i}(\mathbf{b})$ as components of some tensors).

Field variables $a^{k}$ are Casimirs densities. However, the crucial
difference between a Dubrovin--Novikov type first order Hamiltonian
structure and Dubrovin--Novikov type third order Hamiltonian structure is
existence of $n$ extra Casimirs. Such an observation was first made in \cite%
{FGMN} where they were found for the three-component non-diagonalizable
hydrodynamic type system $a_{t}^{1}=a_{x}^{2}$, $a_{t}^{2}=a_{x}^{3}$, $%
a_{t}^{3}=[(a^{2})^{2}-a^{1}a^{3}]_{x}$. Under the potential substitution $%
a^{k}=b_{x}^{k}$ this hydrodynamic type system becomes $b_{t}^{1}=b_{x}^{2}$%
, $b_{t}^{2}=b_{x}^{3}$, $b_{t}^{3}=(b_{x}^{2})^{2}-b_{x}^{1}b_{x}^{3}$.
However, one can choose three new Casimirs $\mathbf{S}^{\beta }=\int
s^{\beta }dx$ such that $s^{1}=b^{1},s^{2}=b^{2},s^{3}=b^{3}+b^{2}b_{x}^{1}$%
. Then the above nonlinear system reduces again to the conservative form $%
s_{t}^{1}=s_{x}^{2}$, $s_{t}^{2}=(s^{3}-s^{2}s_{x}^{1})_{x}$, $%
s_{t}^{3}=(s^{2}s_{x}^{2})_{x}$. Obviously the inverse transformation is $%
b^{1}=s^{1}$, $b^{2}=s^{2}$, $b^{3}=s^{3}-s^{2}s_{x}^{1}$.

Casimirs $\mathbf{S}^{\beta }=\int s^{\beta }dx$ generate \textquotedblleft
zeroth\textquotedblright\ flows, i.e.%
\begin{equation}
0=(g^{is}\partial _{x}+c_{k}^{is}b_{xx}^{k})\frac{\delta \mathbf{S}^{\alpha }%
}{\delta b^{s}}=\phi ^{\beta \gamma }\psi _{\beta }^{i}\partial _{x}\psi
_{\gamma }^{s}\frac{\delta \mathbf{S}^{\alpha }}{\delta b^{s}}.  \label{zero}
\end{equation}

\begin{theorem}
\label{th:vvv} The following functionals determine $n$ nonlocal Casimirs: 
\begin{equation}
\mathbf{S}^{\alpha }=\int \left( \frac{1}{2}\psi _{mk}^{\alpha
}b_{x}^{k}+\omega _{m}^{\alpha }\right) b^{m}dx.  \label{kazimir}
\end{equation}
\end{theorem}

\begin{proof}
Taking into account $\psi _{sm}^{\alpha }=-\psi _{ms}^{\alpha }$,
variational derivatives are (see (\ref{omega})) 
\begin{equation*}
\frac{\delta \mathbf{S}^{\alpha }}{\delta b^{s}}=\psi _{s}^{\alpha }.
\end{equation*}%
Then \textquotedblleft zeroth\textquotedblright\ flows (\ref{zero}) imply%
\begin{equation*}
0=\phi ^{\beta \gamma }\psi _{\beta }^{i}\partial _{x}\psi _{\gamma }^{s}%
\frac{\delta \mathbf{S}^{\alpha }}{\delta b^{s}}=\phi ^{\beta \gamma }\psi
_{\beta }^{i}\partial _{x}\psi _{\gamma }^{s}\psi _{s}^{\alpha }.
\end{equation*}%
Taking into account $\psi _{\gamma }^{s}\psi _{s}^{\alpha }=\delta _{\gamma
}^{\alpha }$, one can see that $\phi ^{\beta \gamma }\psi _{\beta
}^{i}\partial _{x}\delta _{\gamma }^{\alpha }=0$. The Theorem is proved.
\end{proof}

Any Hamiltonian system (\ref{bi}) possesses $n$ conservation laws associated
with Casimirs:%
\begin{equation*}
s_{t}^{\alpha }=\frac{\partial s^{\alpha }}{\partial b^{k}}b_{t}^{k}+\frac{%
\partial s^{\alpha }}{\partial b_{x}^{k}}b_{xt}^{k}=-\left[ \frac{1}{2}\psi
_{mk}^{\alpha }b^{m}\phi ^{\beta \gamma }\psi _{\beta }^{k}\left( \psi
_{\gamma }^{s}\frac{\delta \mathbf{H}}{\delta b^{s}}\right) _{x}+\phi
^{\alpha \gamma }\psi _{\gamma }^{s}\frac{\delta \mathbf{H}}{\delta b^{s}}%
\right] _{x}.
\end{equation*}

\subsection{The Momentum}

\label{sec:momentum}

Any Hamiltonian system (\ref{bi}) has the momentum $\mathbf{P=}\int Pdx$;
this means that\footnote{%
A momentum defines a translation. This
means that we replace the time variable \textquotedblleft $t$%
\textquotedblright\ by the space variable \textquotedblleft $x$%
\textquotedblright\ and simultaneously we replace the Hamiltonian density $h 
$ by the momentum density $P$.} 
\begin{equation*}
b_{x}^{i}=-\phi ^{\beta \gamma }\psi _{\beta }^{i}\partial _{x}\psi _{\gamma
}^{s}\frac{\delta \mathbf{P}}{\delta b^{s}}.
\end{equation*}
Then the momentum density $P$ can be reconstructed, because all variational
derivatives are known: 
\begin{equation}
\frac{\delta \mathbf{P}}{\delta b^{k}}=-\phi _{\beta \gamma }\psi
_{k}^{\beta }\partial _{x}^{-1}\psi _{m}^{\gamma }b_{x}^{m}.  \label{pi}
\end{equation}

\begin{theorem}
\label{th:uuu} The momentum of the Hamiltonian system (\ref{bi}) is 
\begin{equation}
\mathbf{P}=-\int \left( \frac{1}{3}\phi _{\beta \gamma }\omega _{q}^{\beta
}\psi _{pm}^{\gamma }b_{x}^{m}+\frac{1}{2}\phi _{\beta \gamma }\omega
_{p}^{\beta }\omega _{q}^{\gamma }\right) b^{p}b^{q}dx.  \label{pik}
\end{equation}
\end{theorem}

\begin{proof}
Taking into account (\ref{omega}) and $\psi _{sm}^{\alpha }=-\psi
_{ms}^{\alpha }$ variational derivatives (\ref{pi}) reduce to the form%
\begin{equation*}
\frac{\delta \mathbf{P}}{\delta b^{k}}=-\phi _{\beta \gamma }(\psi
_{ks}^{\beta }b_{x}^{s}+\omega _{k}^{\beta })\partial _{x}^{-1}(\psi
_{mp}^{\gamma }b_{x}^{m}b_{x}^{p}+\omega _{m}^{\gamma }b_{x}^{m})=-\phi
_{\beta \gamma }\psi _{ks}^{\beta }\omega _{m}^{\gamma }b^{m}b_{x}^{s}-\phi
_{\beta \gamma }\omega _{k}^{\beta }\omega _{m}^{\gamma }b^{m}.
\end{equation*}%
Indeed, variational derivatives of (\ref{pik}) coincide with above
expressions, where we utilized the right equation from (\ref{seven}). The
Theorem is proved.
\end{proof}

Any Hamiltonian system (\ref{bi}) possesses the conservation law of momentum%
\begin{equation*}
P_{t}=\frac{\partial P}{\partial b^{k}}b_{t}^{k}+\frac{\partial P}{\partial
b_{x}^{k}}b_{xt}^{k}=\left[ b^{m}\omega _{m}^{\beta }\psi _{\beta }^{s}\frac{%
\delta \mathbf{H}}{\delta b^{s}}-\frac{1}{3}\phi _{\beta \gamma }\omega
_{q}^{\beta }b^{q}\psi _{km}^{\gamma }b^{m}\phi ^{\alpha \delta }\psi
_{\alpha }^{k}\left( \psi _{\delta }^{s}\frac{\delta \mathbf{H}}{\delta b^{s}%
}\right) _{x}-Q\right] _{x},
\end{equation*}%
where $Q$ is a local expression of field variables $b^{k}$ and all their
higher derivatives due to well-known formula: $Q_{x}=\frac{\delta \mathbf{H}%
}{\delta b^{s}}b_{x}^{s}$.

\subsection{The Hamiltonian}

\label{sec:hamiltonian}

In this Section we restrict our considerations on Hamiltonian
non-diagonalizable hydrodynamic type systems only (see \eqref{eq:5} and %
\eqref{eq:10}). In Casimir densities $a^{k}$ they are written in the
conservative form (\ref{aa}). Under the potential substitution $%
a^{i}=b_{x}^{i}$ the hydrodynamic type system (\ref{aa}) becomes %
\eqref{eq:57} (cf. (\ref{g}) and (\ref{f})), which leads to (see (\ref{bi}))%
\begin{equation*}
b_{t}^{i}=v^{i}(\mathbf{b}_{x})=-\phi ^{\beta \gamma }\psi _{\beta
}^{i}\partial _{x}\psi _{\gamma }^{s}\frac{\delta \mathbf{H}}{\delta b^{s}}.
\end{equation*}%
Thus (cf. (\ref{pi}))%
\begin{equation}
\frac{\delta \mathbf{H}}{\delta b^{k}}=-\phi _{\beta \gamma }\psi
_{k}^{\beta }\partial _{x}^{-1}\psi _{m}^{\gamma }v^{m}(\mathbf{b}_{x}).
\label{hama}
\end{equation}%
In general, the expressions $\psi _{m}^{\gamma }v^{m}(\mathbf{b}_{x})=(\psi
_{mk}^{\gamma }b_{x}^{k}+\omega _{m}^{\gamma })v^{m}(\mathbf{b}_{x})$ depend
nonlinearly on $b_{x}^{k}$ only. However, in this paper we restrict our
consideration to the case where these expressions depend on $b_{x}^{k}$
\textit{linearly}, i.e. $\psi _{m}^{\gamma }v^{m}(\mathbf{b}_{x})=\eta
_{m}^{\gamma }b_{x}^{m}$, where $\eta _{m}^{\gamma }$ are constant matrices.

\begin{theorem}
  \label{th:aaa} The non-diagonalizable Hamiltonian hydrodynamic type
  systems~\eqref{eq:10} have the Hamiltonians%
\begin{equation}
\mathbf{H}=\frac{1}{2}\int (\zeta _{pqm}b_{x}^{m}-\phi _{\beta \gamma
}\omega _{p}^{\beta }\eta _{q}^{\gamma })b^{p}b^{q}dx,  \label{ham}
\end{equation}%
where $\zeta _{pqm}$ are symmetric constant matrices with respect to first
two indices $p,q$ such that%
\begin{equation}
\zeta _{kpq}=\frac{1}{3}\phi _{\beta \gamma }(\psi _{kp}^{\beta }\eta
_{q}^{\gamma }+2\psi _{qk}^{\beta }\eta _{p}^{\gamma }),  \label{kpq}
\end{equation}%
where the constant matrix $\eta _{q}^{\gamma }$ must satisfy the set of
constraints
\begin{equation}
\phi _{\beta \gamma }(\psi _{qp}^{\beta }\eta _{k}^{\gamma }+\psi
_{kq}^{\beta }\eta _{p}^{\gamma }+\psi _{pk}^{\beta }\eta _{q}^{\gamma
})=0,\quad
\phi _{\beta \gamma }(\omega_{p}^{\beta }\eta_{q}^{\gamma }
- \omega_{q}^{\beta }\eta _{p}^{\gamma }) = 0.
\label{very}
\end{equation}
\end{theorem}

\begin{proof}
Taking into account $\psi _{m}^{\gamma }v^{m}(\mathbf{b}_{x})=\eta
_{m}^{\gamma }b_{x}^{m}$ and (\ref{omega}), the variational derivatives (\ref%
{hama}) reduce to the form
\begin{equation*}
\frac{\delta \mathbf{H}}{\delta b^{k}}=-\phi _{\beta \gamma }\psi
_{kq}^{\beta }\eta _{p}^{\gamma }b^{p}b_{x}^{q}-\phi _{\beta \gamma }\omega
_{k}^{\beta }\eta _{m}^{\gamma }b^{m}.
\end{equation*}%
On the other hand, the variational derivatives of (\ref{ham}) are 
\begin{equation*}
\frac{\delta \mathbf{H}}{\delta b^{k}}=(\zeta _{kpq}-\zeta
_{pqk})b^{p}b_{x}^{q}-
\frac{1}{2}\phi _{\beta \gamma }(\omega _{k}^{\beta }\eta
_{m}^{\gamma } + \omega _{m}^{\beta }\eta
_{k}^{\gamma })b^{m}.
\end{equation*}%
Comparing both the above r.h.s. expressions we obtain the system of linear
algebraic equations 
\begin{equation*}
\zeta _{kpq}-\zeta _{pqk}=-\phi _{\beta \gamma }\psi _{kq}^{\beta }\eta
_{p}^{\gamma },
\end{equation*}%
Taking into account set of constraints\footnote{%
cf. (\ref{very}) and the second set of constraints in (\ref{seven}).} (\ref%
{very}), one can verify that solution (\ref{kpq}) of the above system is
symmetric with respect to the first indices $p,q$. The Theorem is proved.
\end{proof}

\begin{remark}
The solution (\ref{kpq}) is determined up to a total derivative with respect
to \textquotedblleft $x$\textquotedblright . One can remove auxiliary
elements introducing constants $\zeta _{kpq}$ in an alternative more
effective way 
\begin{equation*}
\zeta _{kkm}=0,\text{ \ }\zeta _{kmk}=\zeta _{pkk}=\phi _{\beta \gamma }\psi
_{km}^{\beta }\eta _{k}^{\gamma },\text{ \ }\zeta _{pqk}=\frac{1}{3}\phi
_{\beta \gamma }(\psi _{kq}^{\beta }\eta _{p}^{\gamma }-\psi _{pk}^{\beta
}\eta _{q}^{\gamma }).
\end{equation*}%
\emph{i.e.} the corresponding Hamiltonian $\mathbf{H}$ in such a case
contains less number of terms.
\end{remark}

Any Hamiltonian system (\ref{bi}) possesses the conservation law of energy.
For instance, in the case of hydrodynamic type system (\ref{aa}), one can
obtain (see \eqref{eq:57} and (\ref{ham}) where $\mathbf{H}=\int h(\mathbf{b}%
,\mathbf{b}_{x})dx$)%
\begin{equation*}
h_{t}=\frac{\partial h}{\partial b^{k}}b_{t}^{k}+\frac{\partial h}{\partial
b_{x}^{k}}b_{xt}^{k}=-\left( \frac{\partial h}{\partial b_{x}^{k}}\phi
^{\beta \gamma }\psi _{\beta }^{k}\partial _{x}\psi _{\gamma }^{s}\frac{%
\delta \mathbf{H}}{\delta b^{s}}+\frac{1}{2}g^{ks}(\mathbf{a})\frac{\delta 
\mathbf{H}}{\delta b^{k}}\frac{\delta \mathbf{H}}{\delta b^{s}}\right) _{x}.
\end{equation*}

\section{Second Hamiltonian Structure for the Six-Component WDVV
Hydrodynamic Type System}

\label{sec:second-hamilt-struct}

Reconstructing a third-order Hamiltonian operator $\hat{A}_{2}$ from its
leading term is a very complicated task. Indeed, we know that $\hat{A}_{2}$
is completely determined by its metric coefficients $g^{ik}(\mathbf{a})$ in
the canonical form \eqref{casimir} with respect to its Casimirs. But, if
Casimirs are unknown, then the operator will contain many `spurious' terms
which come from the choice of the coordinate system, like $d_{k}^{ij}$, $%
d_{km}^{ij}$, \dots and whose determination is also a nontrivial task.

However, in our case the `initial' coordinates $a^{k}$ are all hydrodynamic
conservation law densities, i.e. $a^{k}$ depend just on flat coordinates but
not on their higher derivatives. Then we can conjecture that $a^{k}$ are
Casimirs densities for the operator $\hat{A}_{2}$, as it happened in
known examples in $N=3$. This is indeed the case. So, let us interpret the
matrix $g_{is}(\mathbf{u})$ from Corollary~\ref{co:leadingterm} as the
matrix of a covariant pseudo-Riemannian metric.

\begin{theorem}
\label{th:second-hamilt-struct} The metric $g_{is}(\mathbf{u})$ determined
by \eqref{eq:30} is transformed from the coordinates $u^{k}$ to the
coordinates $a^{k}(\mathbf{u})$ as%
\begin{equation*}
g_{ik}(\mathbf{a})=%
\begin{pmatrix}
(a^{4})^{2} & -2a^{5} & 2a^{4} & -(a^{1}a^{4}+a^{3}) & a^{2} & 1 \\ 
-2a^{5} & -2a^{3} & a^{2} & 0 & a^{1} & 0 \\ 
2a^{4} & a^{2} & 2 & -a^{1} & 0 & 0 \\ 
-(a^{1}a^{4}+a^{3}) & 0 & -a^{1} & (a^{1})^{2} & 0 & 0 \\ 
a^{2} & a^{1} & 0 & 0 & 0 & 0 \\ 
1 & 0 & 0 & 0 & 0 & 0%
\end{pmatrix}%
\end{equation*}%
The metric $g_{ik}(\mathbf{a})$ is a Monge metric satisfying Potemin system
\eqref{eq:20}, (\ref{q}) and it generates a Dubrovin--Novikov type
third-order Hamiltonian operator $\hat{A}_{2}$ in canonical form~%
\eqref{casimir}.

The Monge metric $g_{ij}(\mathbf{a})$ admits the decomposition~(\ref{gen}),
where\footnote{In the matrix $P=(\psi^\beta_i)$ the upper index $\beta$ is a
  column index, so that if $\mathbf{\Phi}=(\phi_{\alpha\beta})$ and
  $\mathbf{g}=(g_{ij})$ we can rewrite~(\ref{gen}) as $\mathbf{g}=P\Phi P^T$
  and~\eqref{bi} as  $\mathbf{b}_t = - (P^T)^{-1}\mathbf{\Phi^{-1}}\partial_x
  P^{-1} \delta\mathbf{H}/\delta\mathbf{u}$.}
\begin{equation}
\psi _{i}^{\gamma }=%
\begin{pmatrix}
1 & a^{5} & a^{4} & 0 & 0 & 0 \\ 
0 & a^{3} & 0 & 1 & a^{5} & 0 \\ 
0 & -a^{2} & 0 & 0 & -a^{4} & 1 \\ 
0 & 0 & -a^{1} & 0 & a^{3} & 0 \\ 
0 & -a^{1} & 0 & 0 & -a^{2} & 0 \\ 
0 & 0 & 0 & 0 & -1 & 0%
\end{pmatrix}%
,\quad \phi _{\beta \gamma }=%
\begin{pmatrix}
0 & 0 & 0 & 0 & -1 & 0 \\ 
0 & 0 & 0 & -1 & 0 & 0 \\ 
0 & 0 & 1 & 0 & 0 & 1 \\ 
0 & -1 & 0 & 0 & 0 & 0 \\ 
-1 & 0 & 0 & 0 & 0 & 0 \\ 
0 & 0 & 1 & 0 & 0 & 2%
\end{pmatrix}%
.  \label{eq:37}
\end{equation}%
So the Hamiltonian operator $\hat{A}_{2}$ can be rewritten in the simplified
form (\ref{s}), where inverse matrices are%
\begin{equation*}
\psi _{\gamma }^{i}=\frac{1}{a^{1}}%
\begin{pmatrix}
a^{1} & 0 & 0 & a^{4} & a^{5} & a^{3}a^{4}-a^{2}a^{5}\cr0 & 0 & 0 & 0 & -1 & 
a^{2}\cr0 & 0 & 0 & -1 & 0 & -a^{3}\cr0 & a^{1} & 0 & 0 & a^{3} & 
a^{1}a^{5}-a^{2}a^{3}\cr0 & 0 & 0 & 0 & 0 & -a^{1}\cr0 & 0 & a^{1} & 0 & 
-a^{2} & (a^{2})^{2}-a^{1}a^{4}\cr%
\end{pmatrix}%
,\quad \phi ^{\beta \gamma }=%
\begin{pmatrix}
0 & 0 & 0 & 0 & -1 & 0\cr0 & 0 & 0 & -1 & 0 & 0\cr0 & 0 & 2 & 0 & 0 & -1\cr0
& -1 & 0 & 0 & 0 & 0\cr-1 & 0 & 0 & 0 & 0 & 0\cr0 & 0 & -1 & 0 & 0 & 1%
\end{pmatrix}%
.
\end{equation*}%
In this case $\det \mathbf{\phi }=1$, $\det \mathbf{g}=(a^{1})^{4}$, $\det 
\mathbf{\psi }=-(a^{1})^{2}$.

Each of the hydrodynamic type systems \eqref{eq:10} admits a Hamiltonian
formulation by means of $\hat{A}_{2}$, with nonlocal Hamiltonian density $%
\tilde{h}_{k}(\mathbf{b},\mathbf{b}_{x})$, respectively (see (\ref{ham}))%
\begin{align}
& \tilde{h}
_{1}=-b^{4}b^{5}b_{x}^{1}-b^{5}b^{2}b_{x}^{2}+b^{2}b^{4}b_{x}^{3}-b^{2}b^{6},
\label{eq:28} \\
& \tilde{h}
_{2}=-b^{3}b^{5}b_{x}^{2}+b^{4}b^{3}b_{x}^{3}+b^{1}b^{5}b_{x}^{5}-b^{3}b^{6},
\end{align}%
where $a^{k}=b_{x}^{k}$. Both hydrodynamic type systems \eqref{eq:10} also
have a common momentum $\mathbf{P}=\int Pdx$ and the same set of nonlocal
Casimirs $\mathbf{S}^{k}=\int s^{k}dx$, where (see (\ref{pik}))%
\begin{equation*}
P=-b^{3}b^{2}b_{x}^{2}-b^{1}b^{3}b_{x}^{4}+b^{1}b^{2}b_{x}^{5}-b^{1}b^{6}-(b^{3})^{2},
\end{equation*}%
\begin{equation*}
s^{1}=b^{1},\text{ \ }s^{2}=b^{2},\text{ \ }s^{3}=b^{3},\text{ \ }%
s^{4}=b^{4}b_{x}^{1},\text{ \ }s^{5}=b^{5}b_{x}^{1}+b^{3}b_{x}^{2},\text{ \ }%
s^{6}=b^{5}b_{x}^{2}+b^{3}b_{x}^{4}+b^{6}.
\end{equation*}

The operators $\hat{A}_{1}$, $\hat{A}_{2}$ form a commuting pair, hence the
systems \eqref{eq:10} have a bi-Hamiltonian formulation.
\end{theorem}

\begin{proof}
  It is not difficult (using Reduce) to change the
  coordinates of the metric $g_{ij}(\mathbf{u})$ using the Vi\`{e}te formulae
  \eqref{eq:13}. Then, again by Reduce, it is easy to verify
  that the metric $g_{ij}(\mathbf{a})$ fulfills the condition \eqref{eq:20}
  which ensures the fact that $g_{ij}(\mathbf{a})$ is a Monge metric, and
  fulfills the nonlinear equation~\eqref{eq:21}. This means that the operator
  $\hat{A}_{2}$ defined through \eqref{casimir} is a Hamiltonian operator.

  The metric decomposition (\ref{gen}), (\ref{chet}) follows from the
  integrability of linear system (\ref{d}), which is easily solved by
  Reduce once the linearity of solutions $\psi^i_\gamma$ is
  taken into account.  The nonlocal Casimirs $\mathbf{S}^{k}=\int s^{k}dx$ and
  the momentum $\mathbf{P}=\int Pdx$ can be found by the formulae
  \eqref{kazimir} and \eqref{pik}, respectively.

The computation of Hamiltonians $\mathbf{\tilde{H}}_{k}=\int \tilde{h}_{k}dx$
is slightly more complicated. First of all, one should find constant matrices
$\mathbf{\eta }_{1},\mathbf{\eta }_{2}$ determined by the equations $\psi
_{m}^{\gamma }v^{m}(\mathbf{b}_{x})=\eta _{1m}^{\gamma }b_{x}^{m}$, $\psi
_{m}^{\gamma }w^{m}(\mathbf{b}_{x})=\eta _{2m}^{\gamma }b_{x}^{m}$ (see
\eqref{eq:10}) for both commuting systems $b_{y}^{i}=v^{i}(\mathbf{b}_{x})$ and
$ b_{z}^{i}=w^{i}(\mathbf{b}_{x})$. Here they are
\begin{equation*}
\mathbf{\eta }_{1}=
\begin{pmatrix}
0 & 0 & 0 & 0 & 0 & 0\cr1 & 0 & 0 & 0 & 0 & 0\cr0 & 0 & 0 & 0 & 0 & 0\cr0 & 0
& 0 & 1 & 0 & 0\cr0 & 0 & -2 & 0 & 0 & 1\cr0 & -1 & 0 & 0 & 0 & 0%
\end{pmatrix}
,\quad \mathbf{\eta }_{2}=
\begin{pmatrix}
0 & 0 & 0 & 0 & 0 & 0\cr0 & 0 & 0 & 0 & 0 & 0\cr1 & 0 & 0 & 0 & 0 & 0\cr0 & 0
& 0 & 0 & 0 & 0\cr0 & 0 & 0 & 1 & 0 & 0\cr0 & 0 & -1 & 0 & 0 & 1%
\end{pmatrix}
.
\end{equation*}
The Hamiltonian densities can then be found by the formulae
  \eqref{ham}.

  Finally, it is easy by CDE (but probably impossible by pen and paper) to
  prove that $[\hat{A}_1,\hat{A}_2]=0$, ie the first and third order operators
  commute with respect to the Schouten bracket. This endows both the WDVV
  systems \eqref{eq:10} by the structure of a bi-Hamiltonian system.
\end{proof}

\begin{remark}
Using CDE \cite{cde} it is not difficult to prove that $\hat{A}_{2}(%
\mathbf{p})$ fulfills the criterion of \cite{KeKrVe-JGP-2004} for the left
system \eqref{eq:10}. This means that $\hat{A}_{2}$ maps conservation laws
to symmetries for that system, and analogously for the right system.
\end{remark}

\begin{remark}\label{sec:second-hamilt-struct-1}
  The Monge metric $g_{ij}(\mathbf{a})$ is not flat: its Riemann curvature
  tensor does not vanish. Moreover, its scalar curvature (Ricci
    scalar) is zero\footnote{We are grateful to O. Mokhov for
        communicating us this fact.}  but its sectional curvature is non
    constant. The Monge metric is also not projectively flat: its Weyl tensor
  is nonzero. The known examples in $N=3$ have flat metrics: this is a first
  example of a Dubrovin--Novikov type third-order Hamiltonian operator with a
  non-trivial curvature.
\end{remark}

\section{Conclusion and Outlook}

It is worth to make a short discussion in order to establish to what extent our
results can be checked by a reader.

While the Monge property of $g_{ij}$ is easily checkable by pen and paper, the
hand-verification of the nonlinear condition~\eqref{eq:21} looks like it would
take much longer.

In principle it is possible to find nonlocal Casimirs, the momentum and the
Hamiltonian by pen and paper, using our formulae; of course it is faster to use
computer. However, we stress that the decomposition of the Monge metric and the
fact that $\tilde{h}_1$ and $\tilde{h}_2$ are Hamiltonians for the respective
systems~\eqref{eq:10} can be easily checked by pen and paper, thus validating
our long computer calculations.

The only computation that looks almost impossible to complete by pen and paper
is $[\hat{A}_1,\hat{A}_2]=0$. In this respect CDE proved to yield the
correct result for the Schouten bracket of a fairly high number of test cases
(see \cite{fpv,cde}),
even in the multidimensional situation\footnote{We thank M. Casati for
  providing us examples of commuting and non-commuting bivectors in dimension
  $2+1$.}.

\medskip

Now, let us discuss what are the perspectives of our research work.

Non-diagonalizable hydrodynamic type systems \eqref{eq:5} which possess
conservation law densities quadratic in first derivatives (see (\ref{h}))
were investigated in \cite{fs}. Following the authors of that paper we can
say that in general a non-diagonalizable hydrodynamic type system
\eqref{eq:5} is integrable if it possesses `sufficiently' many conservation
law densities which are quadratic in first derivatives. It was proved that
in the three component case `sufficiently' many means $2$; here, in the four
component case, this number is $3$. In the $n$-component case this question
is open due to its high computational complexity. This problem still is not
yet solved. However, we believe that our Conjecture is valid.

\textbf{Conjecture}: \textit{If a non-diagonalizable hydrodynamic type system
} \eqref{eq:5} \textit{possesses at least one conservation law density which
is quadratic in first derivatives and also has a Dubrovin--Novikov type
first order Hamiltonian structure, then such a system is integrable by the
inverse scattering transform method (this means automatically that such a
system has `sufficiently' enough conservation law densities which are
quadratic in first derivatives). Moreover then this system has a
Dubrovin--Novikov type third order Hamiltonian structure.}

We believe that a consistent subclass of multi-dimensional WDVV associativity
equations written as a family of commuting hydrodynamic type systems are
bi-Hamiltonian in the above sense. A striking feature of all bi-Hamiltonian
examples of WDVV systems is that their third order Hamiltonian operators are
uniquely determined by a quadratic line complex \cite{Dolgachev,fpv}. It would
be interesting to know if this is connected by any means to the rich underlying
geometry (and in particular projective geometry) which is connected to
Gromov--Witten invariants.

In the aforementioned paper \cite{fs}, the authors considered the Matrix
Hopf equation 
\begin{equation*}
U_{t}=(U^{2})_{x},
\end{equation*}
where $U$ is a symmetric matrix of order $N\times N$, with the further
reduction tr$U^{k}=$const, $k=1,2,...,N$. This matrix equation represents a
non-diagonalizable hydrodynamic type system of order $n\times n$ where
$n=N(N-1)/2$. Such a system possesses $N(N+1)/2-1$ hydrodynamic conservation
laws and precisely $N-1$ conservation laws which are quadratic in first
derivatives. This matrix Hopf equation is equivalent to the remarkable $N$ wave
system (see detail in \cite{Ferap3x3}).

Our hypothesis is that the first commuting flows
\begin{equation*}
U_{t^{k}}=(U^{k})_{x},\text{ \ }k=2,3,...,N-1
\end{equation*}%
up to an appropriate reciprocal transformations are equivalent to
corresponding multi-dimensional WDVV associativity equations. This problem
should be investigated elsewhere.

The case considered in this paper is determined by the particular choice $N=4
$. Indeed, our two six-component commuting non-diagonalizable hydrodynamic type
systems have nine hydrodynamic conservation laws and three conservation laws
which are quadratic in first derivatives. They are connected with the pair of
matrix Hopf equations $U_{y}=(U^{2})_{x}$, $ U_{z}=(U^{3})_{x}$ (with four
constraints tr$U^{k}$=const, $k=1,2,3,4$) by an appropriate reciprocal
transformation.

Finally, we remark that the triviality of the Hamiltonian cohomology of
$\hat{A}_{1}$ \cite{getz} implies that $\hat{A}_{2}=L_{\tau }\hat{A}_{1}$.  The
vector field $\tau $ is not uniquely determined; the \emph{mastersymmetry} of
the WDVV associativity equations should be found as one of the possible choices
for $\tau$ \cite{Dorf}. The explicit expression for all possible generalized
vector field $\tau $ of the above type can be of interest, and is computable
via the technique of Lagrangian representation \cite{Yavuz}. We solved this
problem for the $3$-component WDVV \cite{v}; expressions for the $6$-component
WDVV are available upon request.

\section*{Acknowledgements}

We thank M. Casati, B.A. Dubrovin, E.V. Ferapontov, M. Marvan,
O.I. Mokhov and G.V. Potemin for useful discussions. RFV would also
like to thank A.C. Norman for his support with the computer algebra system
Reduce and the system administrator of the server \texttt{sophus} A. Falconieri.

We acknowledge the financial support from GNFM of the Istituto Nazionale
di Alta Matematica, the Istituto Nazionale di Fisica Nucleare, and the
Dipartimento di Matematica e Fisica \textquotedblleft E. De
Giorgi\textquotedblright\ of the Universit\`{a} del Salento. MVP's work was
partially supported by the grant of Presidium of RAS \textquotedblleft
Fundamental Problems of Nonlinear Dynamics\textquotedblright\ and by the RFBR
grant 14-01-00012.

We thank the anonymous Referee for many comments that helped us to improve the
exposition of our results.

\end{document}